\documentclass[12pt, draftclsnofoot, onecolumn]{IEEEtran}

\usepackage{cite}
\usepackage{amsmath,amssymb,amsfonts,bm,color,amsthm}
\usepackage{algorithm,algorithmic}
\usepackage{graphicx}
\usepackage{setspace,multicol,multirow,subfigure,float}
\usepackage{enumitem}
\usepackage{booktabs}

\allowdisplaybreaks[1]

\newtheorem{theorem}{Theorem}
\newtheorem{corollary}{Corollary}

\allowdisplaybreaks[4]

\begin{document}
	
\title{Massive Unsourced Random Access for Near-Field Communications}                  

\author{Xinyu~Xie,
	Yongpeng~Wu,~\IEEEmembership{Senior~Member,~IEEE,}
	Jianping~An,~\IEEEmembership{Senior~Member,~IEEE,}
	Derrick~Wing~Kwan~Ng,~\IEEEmembership{Fellow,~IEEE,}
	Chengwen~Xing,~\IEEEmembership{Member,~IEEE,}
	and~Wenjun~Zhang,~\IEEEmembership{Fellow,~IEEE}
	
	\thanks{X. Xie, Y. Wu, and W. Zhang are with the Department of Electronic Engineering, Shanghai Jiao Tong University, Shanghai 200240, China (e-mails: \{xinyuxie, yongpeng.wu, zhangwenjun\}@sjtu.edu.cn).
		
	J. An and C. Xing are with the School of Information and Electronics, Beijing Institute of Technology, Beijing 100081, China (e-mails: xingchengwen@gmail.com, an@bit.edu.cn).

	D. W. K. Ng is with the School of Electrical Engineering and Telecommunications, University of New South Wales, Sydney, NSW 2052, Australia (e-mail: w.k.ng@unsw.edu.au).

	\emph{Corresponding author: Yongpeng Wu.}}
}

\maketitle
\vspace{-30pt}

\begin{abstract}

This paper investigates the unsourced random access (URA) problem with a massive multiple-input multiple-output receiver that serves wireless devices in the near-field of radiation. We employ an uncoupled transmission protocol without appending redundancies to the slot-wise encoded messages. To exploit the channel sparsity for block length reduction while facing the collapsed sparse structure in the angular domain of near-field channels, we propose a sparse channel sampling method that divides the angle-distance (polar) domain based on the maximum permissible coherence. Decoding starts with retrieving active codewords and channels from each slot. We address the issue by leveraging the structured channel sparsity in the spatial and polar domains and propose a novel turbo-based recovery algorithm. Furthermore, we investigate an off-grid compressed sensing method to refine discretely estimated channel parameters over the continuum that improves the detection performance. Afterward, without the assistance of redundancies, we recouple the separated messages according to the similarity of the users' channel information and propose a modified $ K $-medoids method to handle the constraints and collisions involved in channel clustering. Simulations reveal that via exploiting the channel sparsity, the proposed URA scheme achieves high spectral efficiency and surpasses existing multi-slot-based schemes. Moreover, with more measurements provided by the overcomplete channel sampling, the near-field-suited scheme outperforms its counterpart of the far-field.

\end{abstract}


\section{Introduction}

With the prosperity of the Internet-of-Things (IoT), massive machine-type communications (mMTC)~\cite{DSG17}, also known as massive access~\cite{WGZ20,CNY21}, has been identified as one of the three typical application scenarios in the fifth-generation (5G) mobile communication systems. Specifically, devices in the context of mMTC are activated with low probability and deliver only small-size packets to the base station (BS). Unfortunately, applying conventional access technologies tailored for human-type communications to this emerging scenario would lead to extremely low resource utilization since delicately arranging transmission resources to a massive number of users would result in prohibitive signaling overhead and long latency \cite{ZWZ16}. Therefore, it requires novel theories and paradigms to establish pragmatic machine-centric wireless networks.

In practice, mMTC schemes usually adopt a \emph{grant-free} random access protocol, where the users send their messages directly to the BS without explicit scheduling. As the users are blind to the receiver before communication, a typical class of grant-free massive access schemes, e.g., \cite{CSY18,LY18,KGW20}, rely on the pre-allocation of fixed non-orthogonal pilots to perform joint activity detection and channel estimation (JADCE) before the subsequent data transmission. Due to the sporadic traffic pattern, JADCE can be cast as a compressed sensing (CS) problem and solved by conventional sparse recovery methods such as approximate message passing (AMP)~\cite{CSY18} or orthogonal AMP (OAMP)~\cite{CLP21, CLP23}. However, it becomes difficult or even impossible to assign unique pilots to cope with the rapidly growing number of users, especially in the roll-out of the IoT. To relieve this bottleneck, a novel modality of \emph{unsourced random access} (URA) has been introduced in~\cite{P17}. In particular, URA users are compelled to select their messages from a common codebook. Accordingly, the receiver only acquires a list of transmitted messages independent of the identities of the transmitters. Taking this common codebook assumption, URA has the potential to support infinitely many mMTC users, while the system performance depends only on the number of concurrently active users.

URA was first addressed in~\cite{P17} from the information theory perspective, where the author derived the achievability bound on the per-user error rate in additive white Gaussian noise (AWGN) channels considering the impacts of the finite block length (FBL) theory. Then, comprehensive studies for quasi-static Rayleigh fading channels were conducted in~\cite{KKF20,GWS22,GWL23}. Due to the fact that the dimension of the common codebook grows exponentially with respect to the message length, practical coding schemes approaching the performance of the FBL benchmark take a divide-and-conquer strategy to alleviate the unmanageable computational burden. For instance, a pilot/data division scheme based on interleave division multiple-access (IDMA) was proposed in~\cite{VNC19}. Specifically, in the first stage, pilot sequences are selected from a common pool relying on the first few bits of the payload, which decide the adopted interleavers for the low-density parity check (LDPC) codes to encode the rest of the data. Furthermore, a concatenated coding framework, termed coded compressed sensing (CCS), was presented in~\cite{ACN20}. In particular, each message is partitioned into several segments. Indeed, they are coupled by appending parity check bits generated from a linear block code, thereby forming multiple sub-messages to be sent over successive slots. Meanwhile, the receiver detects the codeword activity and then reconstructs the original messages as the valid paths of a tree-based forward error correction process. Besides, \cite{WZC22} employed a codebook of binary chirps and the shift property was investigated for message stitching. Despite various efforts have been devoted to the literature, the aforementioned schemes are confined to a single-antenna receiver and their results are typically not applicable to multi-antenna systems due to the impacts of fading and the multi-antenna correlation~\cite{WGZ20}.

As a promising technology to support the emerging massive access, massive multiple-input multiple-output (MIMO)~\cite{BHS17} exploits the degrees of freedom (DoF) brought by the spatial dimensions via the same time/frequency resources to mitigate interference and increase spectral efficiency. For instance, pilot-aided URA schemes were extended to the case of massive MIMO in~\cite{FMJ22,LWZ22,GNC22}, where JADCE was formulated as a multiple measurement vector (MMV) problem~\cite{CSY18}. Nevertheless, the CS-based MMV approaches face the fundamental limitation that to detect $ \bar{K}_{\mathrm{a}} $ active codewords among $ \bar{K} $ potential ones, the required coherence block length $ \bar{N} $ is in the order of $ \mathcal{O}( \bar{K}_{\mathrm{a}} \log ( \bar{K} ) ) $ (or empirically, $ \bar{N} > \bar{K}_{\mathrm{a}} $)~\cite{LY18}, which restricts the total coding rate. In fact, considering the channel hardening effect of massive MIMO, the authors in~\cite{FHJ21} introduced a covariance-based estimation formulation to recover large-scale fading coefficients (LSFCs) for the inner codeword activity detection (AD) of CCS. Although the proposed non-Bayesian maximum likelihood estimator in~\cite{FHJ21} can identify up to $ \bar{K}_{\mathrm{a}} = \mathcal{O}( \bar{N}^{2} ) $ simultaneously active users, the required number of antennas must exceed $ \bar{K}_{\mathrm{a}} $, which may not always hold in practice. Recently, an uncoupled compressed sensing (UCS) scheme was proposed in~\cite{SBM21} that leveraged the instantaneous channel state information (CSI) of massive MIMO to improve the coding rate of CCS. Specifically, the strong correlation between users' channels across different time slots was captured to recombine separated codewords, thereby eliminating the need for integrating redundancies in the encoded data. Besides,~\cite{DLG20} proposed to modulate the partitioned data to rank-$ 1 $ tensors and the double-blind data-channel acquisition at the receiver side was conducted by tensor decomposition. It is worth noting that all these works adopted a Rayleigh fading channel model, which is an impractical setting considering that antennas at the BS terminal are subjected to high correlations in realistic wireless propagation environments~\cite{XWG20,GWY21}. On the other hand, more and more experimental evidence and analyses have revealed that the high spatial resolution provided by multiple antennas against finite propagation paths leads to a sparse structure of the physical channels in the angular domain~\cite{BHS10,ZYZ18}. In addition, the authors in~\cite{XWA22} captured the sparsity characteristic to assist the uncoupled URA transmission such that the designed receiver can reduce the block length requested for JADCE and is capable of resolving the codeword collision whenever the arriving directions of conflicting users are separable. Yet, the considered geometric channel model and its sparse representation in~\cite{XWA22} and various massive access schemes, e.g.,~\cite{KGW20,LW21,CZY22,DK23}, rely highly on the plane wave approximation of the spherical wavefront, which is valid in the far-field region of the antenna array where the angular electromagnetic field distribution is independent of the distance.

In fact, the boundary dividing the near- and far-field of radiation, known as the \emph{Rayleigh distance}~\cite{SJ17}, is proportional to the square of the antenna aperture and would be significantly enlarged by the growing number of antennas in massive access scenarios. As a result, a considerable number of users in practical massive MIMO networks would inevitably fall into the near-field area, especially for millimeter-wave/terahertz devices that must be located in proximity to the BS due to inherently severe path loss attenuation~\cite{PZZ22,GTZ22}. Particularly, in practical near-field regions, the spherical wave propagation is strictly obeyed and the arrived wave phase is decided by both the distance and the impinging angle~\cite{CWL23}. Unfortunately, sparsity-exploiting channel estimation (CE) or massive access techniques designed for far-field communications can hardly be adapted to the near-field scenario as the angular domain sparsity under the standard spatial Fourier transformation~\cite{ZYZ18} is now damaged by the essential distance-dependent channel characteristic. Seeking alternative sparse representations for compressive near-field CE, uniform and non-uniform grid partition methods were proposed in~\cite{HJW20} and~\cite{CD22}, respectively. Nevertheless, the resultant channel compression dictionary typically exhibits strong column coherence when applied, which causes non-negligible interference with CS recovery methods. Besides, the dictionary is generally overcomplete, i.e., the number of atoms goes beyond the number of antenna observations, which significantly increases the dimension of the JADCE problem. Moreover, the underlying locations of transmitters need not lie on the discretely sampled grid. The ``off-grid'' effect may further degrade the performance of the reconstruction algorithm. This raises new challenges to the design of massive access schemes for near-field users concerning computational complexity and estimation accuracy.

This paper proposes a URA scheme to support the massive access of near-field users. We develop novel CS algorithms to detect the transmitted codewords and reconstruct the corresponding channels. The recovered CSI further facilitates the redundancy-free message stitching process based on channel grouping. The main contributions of this paper are summarized as follows.
\begin{itemize}
	\item
	We propose a novel polar domain channel sampling method to capture the sparse characteristics of the near-field channel. Motivated by the symmetry observed in the coherence function between near-field array responses, we employ a uniform angle-ring sampling strategy, with intervals jointly determined by the maximum permissible coherence. In comparison to the existing approach for sparse near-field channel representation in \cite{CD22}, the proposed method facilitates lower column coherence in the compressed channel dictionary while maintaining a similar number of channel samples.
	\item
	We exploit the structured sparsity of the near-field channel and propose a novel turbo-type compressive sampling matching pursuit (Turbo-CoSaMP) algorithm to address the JADCE problem. Within the alternating process, we initially consider the coherent sparsity of the near-field channel vector in the spatial domain. The simultaneous CoSaMP (S-CoSaMP) algorithm is introduced for the AD of codewords, which simultaneously narrows the possible range of active elements for CE. To reduce the required measurements of such a MMV approach, we further promote the near-field channel sparsity in the proposed polar domain. We investigate the restricted isometry property (RIP) of the overcomplete channel compression dictionary to develop a two-dimensional (2D) version of CoSaMP for CE, which is computationally more efficient than its equivalent vectorized form. Numerical results suggest that the proposed algorithm is suited for the scenarios where the coherence block length is far less than the number of simultaneously active users, revealing its advantage compared to state-of-the-art sparse recovery methods, e.g.,~\cite{CLP21, FHJ21}.
	\item
	To mitigate the effect of basis mismatch resulting from the discrete channel sampling against the continuity of near-field channel parameters, we embed a Newtonized coordinate descent step inside Turbo-CoSaMP, which jointly optimizes the critical channel parameters, i.e., DoAs and distances, over the continuous domain. A cyclic refinement is further introduced to cooperate with all the intermediate estimations. The developed Newtonized Turbo-CoSaMP (N-Turbo-CoSaMP) algorithm outperforms Turbo-CoSaMP with a fast convergence rate.
	\item
	We design a modified $ K $-medoids method as the kernel of the clustering-based decoding. Both balanced and unbalanced assignment strategies are investigated for efficient message stitching without/with codeword collision. The proposed clustering algorithm updates the cluster center as the centrally located data instances rather than the means of its constituents as in $ K $-means, which is robust to contaminated data and efficient for collision resolution.
\end{itemize}
We demonstrate by simulations that the proposed UCS scheme exhibits a remarkable performance of successful decoding while enjoying high spectral efficiency. Besides, due to the promoted sparsity introduced by the overcomplete dictionary of the polar domain sampling, the designed URA scheme for near-field communications even outperforms its counterpart in far-field cases.

The remainder of this paper is organized as follows. We outline the signal model of the considered UCS scheme for URA in Section II, where the near-field channel model is also introduced. In Section III, we propose a polar domain sampling method to unveil the sparse characteristic of the near-field channel. The Turbo-CoSaMP algorithm and the improved Newtonized method for JADCE are developed in Section IV. In Section V, we propose the modified $ K $-medoids algorithm for clustering decoding in UCS. Numerical results of the system performance are presented in Section VI, followed by concluding remarks drawn in Section VII.

Throughout this paper, we denote the $ j $-th column of matrix $ \mathbf{X} $ by $ \mathbf{x}_{j} $, whose $ i $-th element is represented as $ [ \mathbf{x}_{j} ]_{i} $. $ x_{i, j} $ or $ [ \mathbf{X} ]_{i, j} $ stands for the $ (i, j) $-th entry of $ \mathbf{X} $. Given any complex variable or matrix, $ \Re \{ \cdot \} $ returns its real part. We use superscripts $ (\cdot)^*,(\cdot)^T $, and $ (\cdot)^H $ to denote the conjugate, transpose, and conjugate transpose, respectively. $ \operatorname{vec}( \mathbf{X} ) $ generates a column vector by stacking the columns of $ \mathbf{X} $. $ \operatorname{tr}( \mathbf{X} ) $ and $ \operatorname{det}( \mathbf{X} ) $ calculates the trace and determinant of $ \mathbf{X} $, respectively. The operator $ \otimes $ denotes the Kronecker product and $ \lfloor \cdot \rfloor $ rounds an element to the nearest integer less than or equal to that element. For an integer $ X > 0 $, we use the shorthand notation $ [ X ] $ to represent the set $ \{ 1, 2, \dots, X\} $. $ | \cdot |_c $ denotes the cardinality of a set and $ \mathcal{X} \backslash \mathcal{Y} $ represents the set $ \{z: z \in \mathcal{X}, z \notin \mathcal{Y} \} $. We denote the Euclidean norm of a vector by $ \| \cdot \|_{2} $ and the Frobenius norm of a matrix by $ \| \cdot \|_{F} $. Finally, $ \frac{\partial}{\partial x} $ is the partial derivative with respect to $ x $.

\section{System Model}

\subsection{System Depiction}

\begin{figure}[!t]
	\centering
	\includegraphics[width=11cm]{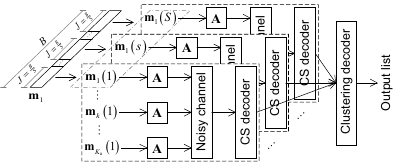}
	\caption{Schematic of the UCS scheme for URA, where the box marked ``$ \mathbf{A} $'' represents the common codebook.}
	\label{ucs}
\end{figure}

We first provide an overview of the UCS transmission protocol employed in this paper before giving specific signal descriptions. At the transmitter side, message sequences are split into small segments, then each of which is mapped to a codeword adopting the common codebook for multi-slot transmission~\cite{SBM21,XWA22}. On the other hand, decoding the transmitted message list takes two steps. The receiver first performs JADCE to the received codewords on a per-slot basis, which is termed CS decoding. Then, codewords or demapped message segments from the same transmitter are reconnected based on the similarity of their corresponding channels, forming the structure of a clustering decoder. The overall architecture of UCS is illustrated in Fig.~\ref{ucs}.

\subsection{Signal Model}

This paper considers a single-cell network where a BS equipped with $ M $ antennas serves $ K_{\mathrm{tot}} $ single-antenna users. In a certain period, only $ K_{\mathrm{a}} \ll K_{\mathrm{tot}} $ mMTC users are activated simultaneously. Under the setting of URA, encoding a $ B $-bit message would require a codebook containing $ 2^{B} $ codewords, which is of high burden even for relatively small-sized data. Therefore, to reduce the system complexity, we split the long message sequence $ \mathbf{m}_{k} \in \{ 0, 1 \}^{B} $, $ k \in [ K_{\mathrm{a}}\textbf{} ] $, into $ S $ small segments, i.e., $ \mathbf{m}_{k} = [\mathbf{m}_{k}( 1 ), \dots, \mathbf{m}_{k}( S )] $, arranged for an $ S $-slot transmission. Each segment of length $ J = B / S $ is then encoded as a codeword of length $ N $ in the dimension-reduced common codebook $ \mathbf{A} = [ \mathbf{a}_{1}, \dots, \mathbf{a}_{2^{J}} ] \in \mathbb{C}^{N \times 2^{J}} $. The encoding process can be interpreted that $ \mathbf{m}_{k}( s ), s \in [ S ] $, is mapped to the codeword $ \mathbf{a}_{\operatorname{dec}( \mathbf{m}_{k}( s ) ) + 1} $, where $ \operatorname{dec}( \cdot ) $ represents the radix ten equivalent of a binary vector. Then, the coded messages are transmitted adopting a coherence block of $ N S $ symbol transmissions, where the channel coefficients remain constant. At the BS, the received signal during the $ s $-th transmission interval is given by
\begin{align}
	\mathbf{Y}( s ) = \sum_{k = 1}^{K_{\mathrm{a}}} \mathbf{a}_{\operatorname{dec}( \mathbf{m}_{k}( s ) ) + 1} \mathbf{h}_{k}^{T} + \mathbf{W}( s ) =  \mathbf{A} \mathbf{\Xi}( s ) \mathbf{H} + \mathbf{W}( s ), \label{sig1}
\end{align}
where $ \mathbf{\Xi}( s ) \in \{ 0, 1 \}^{2^{J} \times K_{\mathrm{a}}} $ is a binary matrix with nonzero values located at the $ ( \operatorname{dec}( \mathbf{m}_{k}( s ) ) + 1 ) $-th entry of the $ k $-th column, $ \mathbf{H} = [ \mathbf{h}_1, \dots, \mathbf{h}_{K_{\mathrm{a}}} ]^{T} \in \mathbb{C}^{K_{\mathrm{a}} \times M} $ is the channel matrix, and $ \mathbf{W}( s ) \in \mathbb{C}^{N \times M} $ is the matrix of independent identically distributed (i.i.d.) AWGN samples with zero mean and variance $ \sigma^{2} $.

\subsection{Channel Model}

An antenna array's field region comprises the near-field and far-field regions. The boundary between them is known as \emph{Rayleigh distance} $ D_{R} = 2 D^{2} / \lambda_{c} $~\cite{SJ17}, where $ D $ is the maximum aperture dimension of the array and $ \lambda_{c} $ is the wavelength of the carrier frequency. In this paper, we employ a half-wavelength-spaced uniform linear array (ULA) with $ M $ antennas equipped at the BS and $ D_{R} = \frac{1}{2} ( M - 1 )^{2} \lambda_{c} $. The physical meaning of Rayleigh distance is that a point source located $ D_{R} $ away from the center of the receiving antenna array would radiate a spherical wavefront such that over the total length $ D $, the wave phases vary less than $ \pi / 8 $ radians~\cite{SJ17}. In other words, the maximum phase error $ \pi / 8 $ allows us to approximate the wavefront by a plane wave. Accordingly, ignoring the subscripts of user indices, the \emph{far-field channel} can be modeled in terms of the underlying physical paths as~\cite{BHS17,CD22}
\begin{align}
	\mathbf{h}_{\mathrm{f}} = \sqrt{\dfrac{M}{L}} \sum_{l = 1}^{L} g_{l} \mathbf{e}_{\mathrm{f}} \left( \theta_{l} \right), \label{fchannel}
\end{align}
where $ L $ is the number of physical paths, $ g_{l} $ is the complex attenuation factor of the $ l $-th path, and we denote by $ \theta_{l} = \sin ( \psi_{l} ) \in ( -1, 1 ) $ the direction of arrival (DoA) with $ \psi_{l} \in ( -\pi / 2, \pi / 2 ) $ the angle of arrival (AoA). Moreover, the array response under the plane wave assumption is given by
\begin{align}
	\mathbf{e}_{\mathrm{f}} \left( \theta_{l} \right) = \frac{1}{\sqrt{M}} \left[ 1, e^{j \pi \theta_{l}}, \cdots, e^{j ( M - 1 ) \pi \theta_{l}} \right]^{T} \in \mathbb{C}^{M}, \label{efar}
\end{align}
which depends on the impinging angle only. Note that $ \mathbf{e}_{\mathrm{f}} $ takes on a Vandermonde structure, i.e., the phase of each element is linear to the antenna index.

Similarly, the multipath \emph{near-field channel} can be modeled as
\begin{align}
	\mathbf{h}_{\mathrm{n}} = \sqrt{\dfrac{M}{L}} \sum_{l = 1}^{L} g_{l} \mathbf{e}_{\mathrm{n}} \left( \theta_{l}, d_{l} \right). \label{nchannel}
\end{align}
However, since the plane wave assumption is no longer valid in the near-field region, the array response $ \mathbf{e}_{\mathrm{n}} ( \theta_{l}, d_{l} ) $ in~\eqref{nchannel} is derived on both the DoA and the distance as~\cite{CD22}
\begin{align}
	\mathbf{e}_{\mathrm{n}} \left( \theta_{l}, d_{l} \right) = \frac{1}{\sqrt{M}} \left[ e^{j \frac{2 \pi}{\lambda_{c}} ( d_{l} - d_{l, 1} )}, \dots, e^{j \frac{2 \pi}{\lambda_{c}} ( d_{l} - d_{l, M} )} \right]^{T} \in \mathbb{C}^{M}, \label{enear}
\end{align}
where $ d_{l} $ is the distance from the $ l $-th scatterer to the center of the antenna array and $ d_{l, m} = \sqrt{d_{l}^{2} - d_{l} \theta_{l} \kappa_{m} \lambda_{c}  + \kappa_{m}^{2} ( \frac{\lambda_{c}}{2} )^{2}} $ is the distance from scatterer to the $ m $-th antenna with $ \kappa_{m} = \frac{2 m - M - 1}{2} $. As can be observed from~\eqref{enear}, the near-field channel exhibits antenna-specific distances, which is highly related to the geolocations of the transmitters.

\section{Sparse Domain Channel Representation}

In this section, we investigate the sparse characteristic of the near-field channel. It is shown that the angular domain channel in the near-field, obtained through spatial Fourier transform, exhibits a severe energy spread. Therefore, we present a joint angle-distance sampling method to achieve a sparse representation of the near-field channel.

\subsection{Limitation of Angular Domain Sampling in the Near-Field}
As the array response in the far-field condition~\eqref{efar} is indeed a discrete Fourier vector, by uniformly sampling DoAs within the range $ ( -1, 1 ) $, we have the equivalent representation that $ \mathbf{h}_{\mathrm{f}} = \mathbf{U} \widetilde{\mathbf{h}}_{\mathrm{f}} $, where $ \widetilde{\mathbf{h}}_{\mathrm{f}} $ is the angular domain channel of $ \mathbf{h}_{\mathrm{f}} $ and
\begin{align}
	\mathbf{U} = \left[ \mathbf{e}_{\mathrm{f}} \left( \tfrac{-M + 1}{M} \right), \dots, \mathbf{e}_{\mathrm{f}} \left( \tfrac{2 ( m - 1 ) - M + 1}{M} \right), \dots, \mathbf{e}_{\mathrm{f}} \left( \tfrac{M - 1}{M} \right) \right] \in \mathbb{C}^{M \times M} \label{dftmtx}
\end{align}
is a discrete Fourier transform (DFT) matrix which decomposes the spatial domain channel into multiple paths along some fixed directions. In fact, the number of physical paths in the realistic wireless communication environment is limited and usually $ L \ll M $ holds, especially in massive MIMO systems~\cite{BHS10,ZYZ18}. Therefore, the angular domain channel vector $ \widetilde{\mathbf{h}}_{\mathrm{n}} $ is sparse with only $ L $ peaks, as shown in Fig.~\ref{channel}.

\begin{figure}[!t]
	\centering
	\includegraphics[width=7cm]{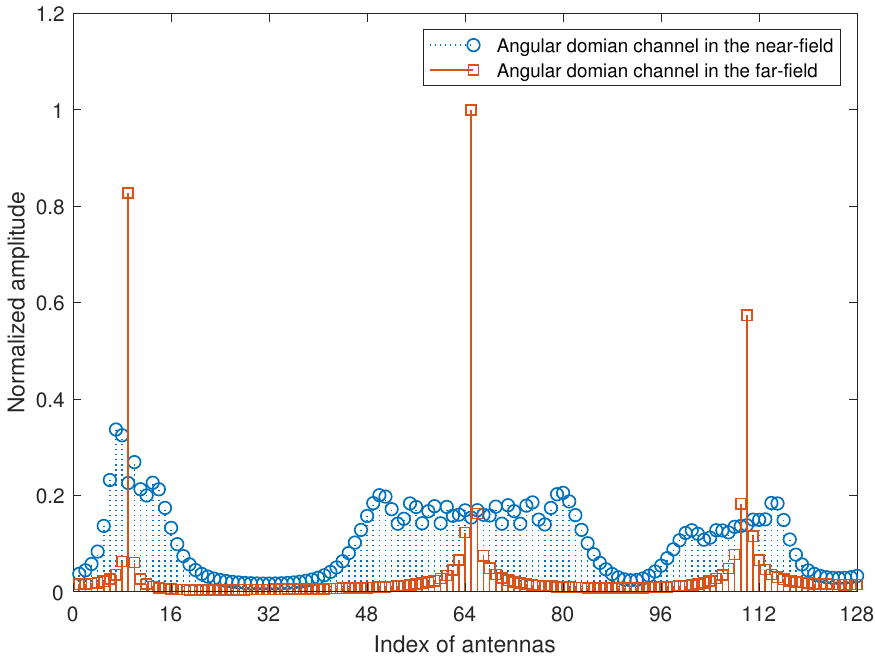}
	\caption{An illustration of the angular domain channel representation in near- and far-field conditions. The considered carrier frequency is $ 3 $ GHz, $ M = 128 $, $ L = 3 $, $ ( \theta_{1}, \theta_{2}, \theta_{3} ) = ( -\sqrt{3} / 2, 0.01, \sqrt{2} / 2 ) $, and $ ( g_{1}, g_{2}, g_{3} ) $ are drawn independently from a complex Gaussian distribution of zero mean and unit variance. The communication distance is set to be $ 1 $ km for the far-field channel and \mbox{$ 10 $ m} for the near-field channel. The maximum amplitude of the channel coefficients is normalized to one.}
	\label{channel}
\end{figure}

In contrast, the Vandermonde structure of~\eqref{efar} vanishes in the near-field condition, leading to a severe energy spread if the angular domain transform is considered. It can be observed in \mbox{Fig.~\ref{channel}} that the energy from a given DoA is captured by the far-field array responses from several other directions. Hence, near-field channels cannot be deemed as sparse in the angular domain.

\subsection{Promoting Near-Field Channel Sparsity via Angle-Distance Sampling}

To reach a specific sparse representation of the near-field channel, sampled array responses should divide the space in both the angle and distance dimensions, i.e., the polar domain. Towards this end, we start with analyzing the coherence between two near-field array responses that can be calculated as \cite{CD22}
\begin{align}
	\Gamma(\theta, d ; \theta', d') &= \left| \mathbf{e}_{\mathrm{n}}^{H}( \theta', d' ) \mathbf{e}_{\mathrm{n}}( \theta, d ) \right| \notag \\
	&\approx \left| \frac{1}{M} \sum_{m = -( M - 1 ) / 2}^{( M - 1 ) / 2} \exp \left( - j m \pi (\theta - \theta') + j \tfrac{\lambda_{c}}{4} m^{2} \pi \left( \tfrac {1-\theta^{2}}{d} - \tfrac{1 - \theta'^{2}}{d'} \right) \right) \right|. \label{cor1}
\end{align}
Define $ \frac{1 - \theta^{2}}{d} = \frac{1}{\phi} $ as a distance ring $ \phi $. In which case, we can equivalently express the coherence function as $ \Gamma(\theta, \frac{1}{\phi} ;\theta', \frac{1}{\phi'}) $. It can be observed in \eqref{cor1} that the coherence between two near-field array responses is determined by the variance of their DoAs and distance rings, i.e., $ \Delta_{\theta} = \theta' - \theta $ and $ \Delta_{\phi} = \frac{1}{\phi'} - \frac{1}{\phi} $. In fact, it can be easily proved that the coherence function possesses the following property
\begin{align}
	\Gamma(\theta, \tfrac{1}{\phi} ; \theta + \Delta_{\theta}, \tfrac{1}{\phi} + \Delta_{\phi}) = \Gamma(\theta, \tfrac{1}{\phi} ; \theta - \Delta_{\theta}, \tfrac{1}{\phi} + \Delta_{\phi}) &= \Gamma(\theta, \tfrac{1}{\phi} ; \theta + \Delta_{\theta}, \tfrac{1}{\phi} - \Delta_{\phi}) \notag \\
	&= \Gamma(\theta, \tfrac{1}{\phi} ; \theta - \Delta_{\theta}, \tfrac{1}{\phi} - \Delta_{\phi}), \label{cor2}
\end{align}
which suggests that we can divide the near-field space by uniformly sampling the angle-ring pairs while keeping the coherence between two adjacent array responses fixed.

In order to determine the intervals $ \Delta_{\theta} $ and $ \Delta_{\phi} $, we approximate \eqref{cor2} using a polynomial function, i.e., $ \Gamma(\theta, \frac{1}{\phi} ; \theta + \Delta_{\theta}, \frac{1}{\phi} + \Delta_{\phi}) \approx 1 - q_{1} M^{2} \Delta_{\theta}^{2} - q_{2} \lambda_{c}^{2} M^{4} \Delta_{\phi}^{2} $, where $ q_{1} = 0.3917 $ and $ q_{2} = 0.001624 $ are multinomial coefficients. With the value of coherence $ \gamma $ fixed, by respectively letting $ \Delta_{\phi} = 0 $ and $ \Delta_{\theta} = 0 $ in the polynomial function, we reach the intervals of angle and distance ring samples, which are represented as
\begin{align}
	\Delta_{\theta} = \frac{1}{M} \sqrt{\frac{1 - \gamma}{q_{1}}}, \quad \Delta_{\phi} = \frac{1}{M^{2} \lambda_{c}} \sqrt{\frac{1 - \gamma}{ q_{2}}}.
\end{align}
Consider that the DoAs are within the range $ ( -1, 1 ) $ and the distances are within the range $ ( D_{F}, + \infty ) $\footnote{Note that the far-field array response is also taken into consideration here by letting $ d = + \infty $.}, where $ D_{F} = 0.5 \sqrt{D^{3} / \lambda_{c}} $ is known as the Fresnel distance \cite{SJ17}. Accordingly, the reciprocals of distance rings are in the range of $ ( 0, 4 \sqrt{2} / ( M \sqrt{M} \lambda_{c} ) ) $. Under such circumstances, the sampled DoAs and distance rings can be expressed as 
\begin{align}
	\tilde{\theta}_{p} &= -1 + \frac{( p - \tfrac{1}{2} )}{M} \sqrt{\frac{1 - \gamma}{q_{1}}}, \ p = 1, \dots, P_{\theta}, \\
	\tfrac{1}{\tilde{\phi}_{p'}} &= \frac{( p' - \tfrac{1}{2} )}{M^{2} \lambda_{c}} \sqrt{\frac{1 - \gamma}{ q_{2}}}, \ p' = 1, \dots, P_{\phi},
\end{align}
where $ P_{\theta} = \lfloor 2 M \sqrt{q_{1} / ( 1 - \gamma ) } \rfloor $ and $ P_{\phi} = \lfloor 4 \sqrt{2 M q_{2} / ( 1 - \gamma )} \rfloor $ are the number of samples in the areas of angle and distance ring, respectively. Furthermore, the sampled distance $ d_{p.p'} $ related to the $ p $-th angle sample and the $ p' $-th distance ring sample can be calculated as $ \tilde{d}_{p.p'} = ( 1 - \tilde{\theta}_{p}^{2} ) \tilde{\phi}_{p'} $. Then, we construct the near-field sampling matrix as $ \mathbf{B} = [ \mathbf{B}_{1}, \dots, \mathbf{B}_{P_{\phi}} ] \in \mathbb{C}^{M \times P} $, where $ \mathbf{B}_{p'} = [ \mathbf{e}_{\mathrm{n}}( \tilde{\theta}_{1}, \tilde{d}_{1, p'} ), \dots, \mathbf{e}_{\mathrm{n}}( \tilde{\theta}_{P_{\theta}}, \tilde{d}_{P_{\theta}, p'} ) ] \in \mathbb{C}^{M \times P_{\theta}} $ and $ P = P_{\theta} P_{\phi} $. It is worth mentioning that we have $ P > P_{\theta} > M $, i.e., $ \mathbf{B} $ is an overcomplete matrix. Finally, the near-field channel $ \mathbf{h}_{\mathrm{n}} $ can be transformed into its polar domain representation $ \widetilde{\mathbf{h}}_{\mathrm{n}} $ by $ \mathbf{h}_{\mathrm{n}} = \mathbf{B} \widetilde{\mathbf{h}}_{\mathrm{n}} $.

Note that a different polar domain sampling method was developed in \cite{CD22}. The authors uniformly sample $ M $ DoAs within $ ( -1, 1 ) $ as a precondition and then compute the distances based on the maximum coherence controlled by a parameter $ \beta_{\Delta} $ and the minimum distance $ \rho_{\mathrm{min}} $. Letting $ \rho_{\mathrm{min}} = D_{F} $, the total number of samples is given by $ M \lfloor \sqrt{M/2} / \beta_{\Delta} \rfloor $. In contrast, in the proposed approach, the DoAs and distances/distance rings are jointly determined based on parameter $ \gamma $. Assuming an equal maximum permissible coherence for both approaches, the proposed method yields a larger number of samples due to oversampling in DoAs. In other words, given a similar number of samples, the channel dictionary constructed adopting the proposed method exhibits lower maximum column coherence compared to that obtained in \cite{CD22}.

It is indicated in~\cite{BHS10,ZYZ18} that exploiting the sparsity of multipath channels in certain transformation domains promotes the DoF with respect to CSI acquisition to $ K_{\mathrm{a}} ( 1 - c \frac{K_{\mathrm{a}}}{N} ) $ in massive MIMO systems, where $ c \in ( 0, 1 ) $ is a parameter proportionally to the channel sparsity $ \rho $. For the far-field channel modeled as~\eqref{fchannel} with $ M $ DoA resolutions, we have $ \rho = \frac{| \mathcal{S} |_{c}}{M} $~\cite{ZYZ18}, where $ \mathcal{S} = \{ m : | [ \widetilde{\mathbf{h}}_{\mathrm{f}} ]_{m} | > \epsilon \} $ is the set of indices of dominant entries in $ \widetilde{\mathbf{h}}_{\mathrm{f}} $ for some appropriately chosen $ \epsilon $. In contrast, the polar domain sampling offers $ P \gg M $ spatial resolutions. Therefore, exploiting the polar domain of near-field channels provides more available DoF for design that can potentially enhance the recovery performance.

\section{2D-CS-Based Codeword Activity Detection and Channel Estimation}

In this section, we aim to determine the codeword activity and recover the corresponding channels. We first take both the spatial row sparsity and the polar domain sparsity into account to propose the Turbo-CoSaMP algorithm, operating S-CoSaMP for AD and 2D-CoSaMP for CE in a turbo manner. Newtonized refinements are further investigated within the turbo detection process to suppress the influence of basis mismatch caused by the discrete channel sampling.

\subsection{Problem Formulation}

Replacing $ \mathbf{H} $ in~\eqref{sig1} by the equivalent polar domain channel expression, we represent the received signal from the near-field users as
\begin{align}
	\mathbf{Y} =  \mathbf{A} \underbrace{\mathbf{\Xi} \widetilde{\mathbf{H}}_{\mathrm{n}} \mathbf{B}^{T}}_{\triangleq \mathbf{Z}} + \mathbf{W} = \mathbf{A} \mathbf{X} \mathbf{B}^{T} + \mathbf{W}, \label{mtxsig}
\end{align}
where $ \widetilde{\mathbf{H}}_{\mathrm{n}} = [ \widetilde{\mathbf{h}}_{\mathrm{n} 1}, \dots, \widetilde{\mathbf{h}}_{\mathrm{n} K_{\mathrm{a}}} ]^{T} \in \mathbb{C}^{K_{\mathrm{a}} \times P} $, $ \mathbf{Z} \triangleq \mathbf{\Xi} \widetilde{\mathbf{H}}_{\mathrm{n}} \mathbf{B}^{T} = \mathbf{X} \mathbf{B}^{T} \in \mathbb{C}^{2^{J} \times M} $, and $ \mathbf{X} \triangleq \mathbf{\Xi} \widetilde{\mathbf{H}}_{\mathrm{n}} \in \mathbb{C}^{2^{J} \times P} $. To detect the codeword activity patterns and recover the relevant channels is equivalent to reconstructing either $ \mathbf{Z} $ or $ \mathbf{X} $ from the noisy observation $ \mathbf{Y} $. Note that we do not aim to distinguish $ \mathbf{\Xi} $ and $ \widetilde{\mathbf{H}}_{\mathrm{n}} $ from $ \mathbf{X} $, since the URA receiver has no obligation to relate the codeword to its source~\cite{P17,WGZ20}, i.e., the users are unsourced from the perspective of the BS. Also, the circumstance of codeword collision is considered in this section, i.e., the columns of $ \mathbf{\Xi} $ might have more than one nonzero element.

In fact, the targeted signal reveals a two-tier sparsity structure. First, due to the sporadic traffic pattern of the mMTC users, the matrix $ \mathbf{Z} $ is row sparse and elements within each row share the same sparsity profile. Moreover, each polar domain channel vector has a small portion of entries with significant amplitude~\cite{CD22}, which further decreases the sparse rate of the model.

Various existing works on massive access, e.g., \cite{CSY18,FMJ22,SBM21}, capture the row sparsity of $ \mathbf{Z} $ and address the JADCE problem by MMV approaches. Indeed, the asymptotic analysis of the MMV approach has proven that the AD error rate tends to be zero in the asymptotic limit where $ K_{\mathrm{a}} $, $ 2^{J} $, $ N $, and $ M $ all tend to infinity, while $ 2^J / N $ and $ K_{\mathrm{a}} / 2^{J} $ converge to some fixed positive values~\cite{LY18}. Nevertheless, it generally requires that $ N > K_{\mathrm{a}} $ holds. To decrease the required number of measurements, \cite{KGW20,XWA22} further leverage the angular domain sparsity of the far-field massive MIMO channel. In practice, they eliminate the effects of the DFT-based angular domain transformation matrix by simply multiplying its conjugate transpose, which reduces the corresponding JADCE problem to a standard formulation of sparse matrix inversion. However, in the near-field case, the approach is ineffective as $ \mathbf{B} $ is an overcomplete matrix. In the remainder of this section, we will develop novel JADCE algorithms for near-field massive access, which take both the row sparsity and the polar domain channel sparsity into consideration while handling the overcomplete near-field channel dictionary.

\subsection{Proposed Tubro-CoSaMP Algorithm for Signal Reconstruction}

We turn to the framework of CoSaMP~\cite{NT09} to handle the preceding sparse recovery problem, which takes a sequential selective strategy similar to the grouped successive interference cancellation (SIC) approach. Specifically, it follows an iterative procedure that retrieves prominent signal components by treating all the others as interferences and then removes the contribution of the newly estimated ones from the residual signal. This choice is motivated by the fact that CoSaMP, being a greedy pursuit approach, yields exact solutions with guaranteed speed under well-defined conditions. However, initially designed for the single measurement vector (SMV) problem, CoSaMP would have to handle a sensing matrix of unmanageable dimensions under the equivalent vectorized form of~\eqref{mtxsig}. Also, it assumes an independent signal support profile and ignores the underlying structured sparsity across elements.

\subsubsection*{Codeword AD via S-CoSaMP}

To fully exploit the structured sparsity, we first capture the coherent support pattern with respect to the spatial domain elements of near-field channels for AD. In fact, concurrently reconstructing signals sharing a common sparsity profile can be cast as a MMV problem~\cite{CSY18,LY18,KGW20} or a simultaneous sparse approximation problem~\cite{TGS06}. We introduce S-CoSaMP, which is the simultaneous version of CoSaMP, as the solution. Specifically, S-CoSaMP handles the spatial correlation in $ \mathbf{Z} $, or its residual $ \mathbf{R} $, by jointly calculating the coherence between each row and the sensing matrix, i.e., $ \mathbf{P}_{\mathrm{r}} = \mathbf{A}^{H} \mathbf{R} $. We determine the row support of $ \mathbf{Z} $ by an energy detection for the rows of $ \mathbf{P}_{\mathrm{r}} $ and elect $ 2 K_{\mathrm{a}} $ rows\footnote{In practical systems, the value of $ K_{\mathrm{a}} $ is random and generally unknown to the receiver. This can be addressed by replacing $ K_{\mathrm{a}} $ with the knowledge of the maximum number of active users $ K_{\mathrm{max}} $~\cite{DLG20}.
	
The idea of simultaneously choosing $ 2 K_{\mathrm{a}} $ active candidates in CoSaMP comes from the combinatorial algorithm~\cite{GST07} that accelerates the convergence speed.} with the largest energy. Then, their indices are combined with the outcome from the former stage to form the index set $ \mathcal{K}_{\mathrm{r}} $, which is a provisional list of the potential active codewords. Note that it still requires $ N > K_{\mathrm{a}} $ to identify active codewords with high probability. At this point, introducing the polar domain sparse structure would further promote the sparsity level of the CS paradigm and reduce the number of channel coefficients to be estimated, hoping to decrease the required block length.

\subsubsection*{Near-Field CE via 2D-CoSaMP}

After acquiring the codeword activity, we further leverage the polar domain sparsity for near-field CE, with the purpose of reconstructing the sparse channel matrix based on the following model
\begin{align}
	\mathbf{R} = \mathbf{A}_{\mathcal{K}_{\mathrm{r}}} \mathbf{X}_{\mathcal{K}_{\mathrm{r}}, :} \mathbf{B}^{T} + \widetilde{\mathbf{W}}, \label{mtxsig1}
\end{align}
where $ \mathbf{A}_{\mathcal{K}_{\mathrm{r}}} $ and $ \mathbf{X}_{\mathcal{K}_{\mathrm{r}}, :} $ are constructed by placing columns $ \mathbf{a}_{j}, j \in \mathcal{K}_{\mathrm{r}} $ and rows $ \mathbf{x}_{j, :}, j \in \mathcal{K}_{\mathrm{r}} $ in order, respectively, and $ \widetilde{\mathbf{W}} $ is the equivalent noise matrix. A possible solution starts with equivalently writing~\eqref{mtxsig1} in the vector form as $ \boldsymbol{r} = \boldsymbol{A} \boldsymbol{x} + \boldsymbol{w} $, where $ \boldsymbol{r} = \operatorname{vec}( \mathbf{R} ) \in \mathbb{C}^{N M} $, $ \boldsymbol{A} = \mathbf{B} \otimes \mathbf{A}_{\mathcal{K}_{\mathrm{r}}} \in \mathbb{C}^{N M \times | \mathcal{K}_{\mathrm{r}} |_{c} P} $, $ \boldsymbol{x} = \operatorname{vec}( \mathbf{X}_{\mathcal{K}_{\mathrm{r}}, :} ) \in \mathbb{C}^{| \mathcal{K}_{\mathrm{r}} |_{c} P} $, and $ \boldsymbol{w} = \operatorname{vec}( \widetilde{\mathbf{W}} ) \in \mathbb{C}^{N M} $. The resultant one-dimensional (1D) CS recovery problem can be tackled by various conventional SMV approaches, e.g., orthogonal matching-pursuit (OMP)~\cite{TG07} and AMP~\cite{DMM09}. However, this exponentially increases the computational complexity considering the large-scale signal involved in massive access.

To overcome the high computational burden, we stick to the matrix sketch~\eqref{mtxsig1} to retrieve the two-dimensional (2D) signal $ \mathbf{X}_{\mathcal{K}_{\mathrm{r}}, :} $ from $ \mathbf{R} $ with known measurement matrices $ \mathbf{A}_{\mathcal{K}_{\mathrm{r}}} $ and $ \mathbf{B} $. As the preliminary of the following discussion, we introduce the RIP~\cite{RV06} to characterize two random measurement matrices $ \widetilde{\mathbf{A}} \in \mathbb{C}^{N' \times J'} $ and $ \widetilde{\mathbf{B}} \in \mathbb{C}^{M' \times P'}  $. In particular, we say that $ \widetilde{\mathbf{A}} $ and $ \widetilde{\mathbf{B}} $ satisfy the RIP of order $ r $ if there exists a constant $ \delta_{r} $ such that
\begin{align}
	\left( 1 - \delta_{r} \right) \| \widetilde{\mathbf{X}} \|_{F}^{2} \leqslant \| \widetilde{\mathbf{A}} \widetilde{\mathbf{X}} \widetilde{\mathbf{B}}^{T} \|_{F}^{2} \leqslant \left( 1 + \delta_{r} \right) \| \widetilde{\mathbf{X}} \|_{F}^{2}, \label{rip}
\end{align}
holds for all $ \widetilde{\mathbf{X}} \in \mathbb{C}^{J' \times P'} $ that is $ r $-sparse, i.e., it has at most $ r $ nonzero entries. Suppose that $ \delta_{r} \ll 1 $, $ \mathbf{X}_{\mathrm{a}} = \widetilde{\mathbf{A}}^{H} \widetilde{\mathbf{A}} \widetilde{\mathbf{X}} \widetilde{\mathbf{B}}^{T} \widetilde{\mathbf{B}}^{*} $ approximates $ \widetilde{\mathbf{X}} $ by the reason that the $ j' $-th column of $ \widetilde{\mathbf{A}} $ and the $ p' $-th column of $ \widetilde{\mathbf{B}} $ capture most of the energy of $ [ \widetilde{\mathbf{X}} ]_{j', p'} $. In words, the large elements in $ \mathbf{X}_{\mathrm{a}} $ point to the nonzero entries in $ \widetilde{\mathbf{X}} $. In Appendix A, we provide the RIP analysis with respect to a partial DFT matrix $ \mathbf{A} $ and the proposed channel dictionary $ \mathbf{B} $, which suggests that increasing the number of measurements or the number of antennas allows a lower threshold of $ \delta_{r} $.

Denote $ R $ the sparsity level of $ \mathbf{X} $. According to the RIP, we recognize $ 2 R $ atoms that contribute most to the current residual by selecting the $ 2 R $ largest entries from the proxy $ \mathbf{A}_{\mathcal{K}_{\mathrm{r}}}^{H} \mathbf{RB}^{*} $. Note that the involved matrix-to-matrix multiplications require $ \mathcal{O}( N M ( P + |\mathcal{K}_{\mathrm{r}} |_{c} ) ) $ computations, while the equivalent calculation in the case of 1D-CS would yield a much higher complexity of $ \mathcal{O}( N M P | \mathcal{K}_{\mathrm{r}} |_{c} ) $. Then, the indices of these atoms, together with the current estimated signal support, are merged to a set $ \mathcal{K} = \{ ( j_{1}, p_{1} ), \dots, ( j_{K}, p_{K} ) \} $ with $ | \mathcal{K} |_{c} = K  $, pointing to possible active elements of $ \mathbf{X} $. The estimations of the corresponding channel coefficients manifest as the arguments $ \widetilde{\mathbf{x}} = [ \widetilde{x}_{j_{1}, p_{1}}, \dots, \widetilde{x}_{j_{K}, p_{K}} ]^{T} \in \mathbb{C}^{K} $ that reduce most of the residual
\begin{align}
	\mathbf{R} = \mathbf{Y} - \sum_{k = 1}^{K} \widetilde{x}_{j_{k}, p_{k}} \mathbf{a}_{j_{k}} \mathbf{b}_{p_{k}}^{T},
\end{align}
which can be calculated as the least squares (LS) solution
\begin{align}
	\arg \min \limits_{\widetilde{\mathbf{x}}} \| \mathbf{R} \|_{F}^{2} = \operatorname{tr} ( \mathbf{R}^{H} \mathbf{R} ) = \| \mathbf{Y} \|_{F}^{2} - \mathbf{n}^{H} \widetilde{\mathbf{x}} - \mathbf{n}^{T} \widetilde{\mathbf{x}}^{*} + \widetilde{\mathbf{x}}^{H} \mathbf{\Lambda} \widetilde{\mathbf{x}}, \label{trace}
\end{align}
where $ \mathbf{n} = \left[ \mathbf{b}_{p_{1}}^{H} \mathbf{Y}^{T} \mathbf{a}_{j_{1}}^{*}, \dots, \mathbf{b}_{p_{K}}^{H} \mathbf{Y}^{T} \mathbf{a}_{j_{K}}^{*} \right]^{T} $ and the $ (k, k') $-th element of $ \mathbf{\Lambda} \in \mathbb{C}^{K \times K } $ can be expressed as $ \left[ \mathbf{\Lambda} \right]_{k, k'} = ( \mathbf{a}_{j_{k}}^{H} \mathbf{a}_{j_{k'}} ) ( \mathbf{b}_{p_{k}}^{H} \mathbf{b}_{p_{k'}} ) $. The optimal $ \widetilde{\mathbf{x}} $ that minimizes $ \| \mathbf{R} \|_{F}^{2} $ is certainly the value of $ \widetilde{\mathbf{x}} $ when the derivative of~\eqref{trace} equals to zero as the function is convex, i.e., $ \frac{\partial \operatorname{tr} \left( \mathbf{R}^{H} \mathbf{R} \right)}{\partial \widetilde{\mathbf{x}}} = \mathbf{\Lambda}^{H} \widetilde{\mathbf{x}} - \mathbf{n} = \mathbf{0} $. Noticing that $ \mathbf{\Lambda} = \mathbf{\Lambda}^{H}  $, the solution is given by $ \widetilde{\mathbf{x}} = \left( \mathbf{\Lambda} \right)^{-1} \mathbf{n} $.

However, only some of the intermediate estimations are retained for residua renewal. First, based on the AD criterion of S-CoSaMP, only $ K_{\mathrm{a}} $ rows of the estimated $ \mathbf{X}_{\mathcal{K}^{\mathrm{tur}}, :} $ are determined active. Then, 2D-CoSaMP selects at most $ R $ significant components among these rows to estimate $ \mathbf{X} $. Subsequently, the residual is renewed by removing the selected atoms from $ \mathbf{Y} $ to end one iteration. Based on the intuition of SIC, as the impacts of the stronger elements are subtracted from the received signal, we hope that the weaker elements can be retrieved from the residual in the next iteration.

\begin{algorithm}[!t]
	\caption{Turbo-CoSaMP}
	\label{turbocosamp}
	\begin{spacing}{1.2}
		\textbf{Input:} Observation $ \mathbf{Y} $, measurement matrices $ \mathbf{A}, \mathbf{B} $, sparsity levels $ K_{\mathrm{a}}, R $, target power $ \tau $ \\
		\textbf{Initialize:} $ \widehat{\mathbf{X}}( 0 ) = \mathbf{0}_{2^{J} \times P} $, row support set $ \mathcal{R}_{\mathrm{r}}( 0 ) = \emptyset $, signal support set $ \mathcal{R}( 0 ) = \emptyset $, residual $ \mathbf{R}( 0 ) = \mathbf{Y} $, $ t = 0 $
		\begin{algorithmic}[1]
			\REPEAT
			\STATE \textbf{\% AD via S-CoSaMP}
			\STATE Compute the signal proxy $ \mathbf{P}_{\mathrm{r}}( t ) = \mathbf{A}^{H} \mathbf{R}( t - 1 ) $
			\STATE Identify the row support set $ \mathcal{S}_{\mathrm{r}}( t ) = \arg \max \limits_{\mathcal{S}_{\mathrm{r}} \subseteq [ 2^{J} ], | \mathcal{S}_{\mathrm{r}} |_{c} = 2 K_{\mathrm{a}}} \sum_{j \in \mathcal{S}_{\mathrm{r}}} \| \mathbf{p}_{\mathrm{r}_{j, :}}( t ) \|_{2}^{2} $
			\STATE Merge the row support set $ \mathcal{K}_{\mathrm{r}}( t ) = \mathcal{R}_{\mathrm{r}}( t - 1 ) \cup \mathcal{S}_{\mathrm{r}}( t ) $
			\STATE \textbf{\% CE via 2D-CoSaMP}
			\STATE Compute the signal proxy $ \mathbf{P}( t ) = \mathbf{A}_{\mathcal{K}_{\mathrm{r}}( t )}^{H} \mathbf{R}( t - 1 ) \mathbf{B}^{*} $
			\STATE Identify the signal support set $ \mathcal{S}( t ) = \arg \max \limits_{\mathcal{S} \subseteq \mathcal{F}, | \mathcal{S} |_{c} = 2 R} \sum_{( j, p ) \in \mathcal{S}} | p_{j, p}( t ) | $, where $ \mathcal{F} = \{ ( j, p ):j \in [ 2^{J} ], p \in [ P ] \} $
			\STATE Merge the signal support set $ \mathcal{K}( t ) = \mathcal{R}( t - 1 ) \cup \mathcal{S}( t ) $
			\STATE LS estimation $ \widetilde{\mathbf{x}}( t ) = \arg \min \limits_{\widetilde{\mathbf{x}}} \| \mathbf{Y} - \sum_{( j, p ) \in \mathcal{K}( t )} \widetilde{x}_{j, p} \mathbf{a}_{j} \mathbf{b}_{p}^{T} \|_{F}^{2} $
			\STATE Screen $ \mathcal{K}_{\mathrm{r}}( t ) $ by retaining the indices of $ K_{\mathrm{a}} $ rows with the largest energy to obtain the row support set $ \mathcal{R}_{\mathrm{r}}( t ) $
			\STATE Screen $ \mathcal{K}( t ) $ by retaining the indices of the $ R $ largest $ | \widetilde{x} | $ within the rows indexed by $ \mathcal{R}_{\mathrm{r}}( t ) $ to obtain the signal support set $ \mathcal{R}( t ) $
			\STATE $ \forall ( j, p ) \in \mathcal{R}( t ) $: $ \widehat{x}_{j, p}( t ) = \widetilde{x}_{j, p}( t ) $
			\STATE $ \forall ( j, p ) \in \mathcal{F} \setminus \mathcal{R}( t ) $: $ \widehat{x}_{j, p}( t ) = 0 $
			\STATE Update the residual $ \mathbf{R}( t ) = \mathbf{Y} - \sum_{( j, p ) \in \mathcal{R}( t )} \widehat{x}_{j, p} \mathbf{a}_{j} \mathbf{b}_{p}^{T} $
			\UNTIL $ \| \mathbf{R}( t ) \|_{F}^{2} \leqslant \tau^{2} $ 
		\end{algorithmic}
		\textbf{Output:} Signal estimation $ \widehat{\mathbf{X}}( t ) $
	\end{spacing}
\end{algorithm}

In general, we arrange S-CoSaMP for AD and 2D-CoSaMP for CE to work in a turbo manner, and halt the iteration when the residual power is reduced to a certain value $ \tau $. The corresponding detection scheme, termed \emph{Turbo-CoSaMP}, is summarized in Algorithm~\ref{turbocosamp}. Specifically, we first introduce S-CoSaMP to perform AD by energy detection based on the correlation computation between each row and the sensing matrix (see line~3, Algorithm~\ref{turbocosamp}), which leverages the shared support across elements of the near-field spatial domain channel. Then, 2D-CoSaMP leverages the near-field channel sparsity in the polar domain, which is represented by the proposed overcomplete channel dictionary in Section III, to reduce the scale of the CE problem. In particular, the sparse channel elements are identified based on a correlation calculation from the matrix sketch to manage the computational complexity (see line~7, Algorithm~\ref{turbocosamp}), which motivates from the RIP analysis of the overcomplete dictionary. We cooperate S-CoSaMP and 2D-CsSaMP to work iteratively, such that S-CoSaMP offers a fast selection strategy for 2D-CoSaMP to narrow the possible range of active coordinates, and the latter captures the extra sparsity in the polar domain channel to reduce the required block length for effective detection. This leads to systematic improvement in JADCE compared with applying either S-CoSaMP or 2D-CoSaMP independently.

Finally, we set an appropriate threshold $ \upsilon $ and make a hard decision on the support of codewords for AD. The index set of active codewords is given by
\begin{align}
	\mathcal{A}( s ) = \left\lbrace j : \left\| \widehat{\mathbf{x}}_{j, :}( s ) \right\|_{2}^{2} > \upsilon, j \in \left[ 2^{J} \right] \right\rbrace \label{L}, s \in [ S ],
\end{align}
where $ \widehat{\mathbf{x}}_{j, :}( s ) $ is the $ j $-th row of estimation $ \widehat{\mathbf{X}} $ at the $ s $-th slot.

\subsubsection{Complexity Analysis}

We analyze the computational complexity of Turbo-CoSaMP within one iteration. In general, support identification and fusion can be computed in $ \mathcal{O}( L_{r} ) $ time with $ L_{r} $ denoting the size of the search space. With respect to S-CoSaMP, calculating the proxy $ \mathbf{P}_{\mathrm{r}} $ yields the complexity of $ \mathcal{O}( 2^{J} N M ) $. As for 2D-CoSaMP, the proxy $ \mathbf{P} $ is computed with complexity $ \mathcal{O}( N M ( K_{\mathrm{a}} + P ) ) $.  Then, generating vector $ \mathbf{n} $ and matrix $ \mathbf{\Lambda} $ require $ \mathcal{O}( R ( N M + N ) ) $ and $ \mathcal{O}( R^{2} ( N + M ) ) $ computations, respectively, and calculating the inverse of $ \mathbf{\Lambda} $ and matrix-to-vector multiplication $ \mathbf{\Lambda}^{-1} \mathbf{n}^{*} $ takes $ \mathcal{O}( R^{3} ) $ and $ \mathcal{O}( R^{2} ) $ operations, respectively. Hence, the overall cost of the LS estimation is $ \mathcal{O}( R N M ) $. Finally, the residual is renewed with complexity $ \mathcal{O}( R N M ) $. Empirical observations suggest that the matrix-to-matrix multiplications involved in the proxy calculation dominate the overall complexity of Turbo-CoSaMP, which is $ \mathcal{O}( N M \max \{ 2^{J}, K_{\mathrm{a}} + P \} ) $.

\subsubsection{Convergence Guarantee}

We present the following theorem to claim that the stopping criterion leads to a guarantee about the mean-squared error (MSE) of the final signal estimation.
\begin{theorem}
	Let $ \mathbf{X} $ be a $ R $-sparse matrix and $ \widehat{\mathbf{X}} $ is the $ R $-sparse approximation of $ \mathbf{X} $. Let $ \mathbf{Y} = \mathbf{AXB}^{T} + \mathbf{W} \in \mathbb{C}^{N \times M} $ where each entry of $ \mathbf{W} $ is generated independently from the distribution $ \mathcal{CN}( 0, \sigma^{2} ) $. $ \mathbf{R} = \mathbf{Y} - \mathbf{A} \widehat{\mathbf{X}} \mathbf{B}^{T} $ is the residual. Assume that the measurement matrices $ \mathbf{A} $ and $ \mathbf{B} $ satisfy the RIP with $ \delta_{4 R} \leqslant \varepsilon $. When the stopping criterion is satisfied, it holds that
	\begin{align}
		\| \widehat{\mathbf{X}} - \mathbf{X} \|_{F} \leqslant \dfrac{\tau + \sqrt{N M} \sigma}{\sqrt{1 - \varepsilon}}. \label{theorem1}
	\end{align}
\end{theorem}
\begin{proof}
	The matrix $ \widehat{\mathbf{X}} - \mathbf{X} $ is $ 2 R $-sparse, i.e., it has at most $ 2 R $ nonzero entries. Thus, we have
	\begin{align}
		\sqrt{1 - \delta_{2 R}} \| \widehat{\mathbf{X}} - \mathbf{X} \|_{F} \overset{(\mathrm{a})}{\leqslant} \| \widetilde{\mathbf{A}} ( \widehat{\mathbf{X}} - \mathbf{X} ) \widetilde{\mathbf{B}}^{T} \|_{F} = \| \mathbf{R} - \mathbf{W} \|_{F} \leqslant \| \mathbf{R} \|_{F} + \| \mathbf{W} \|_{F},
	\end{align}
	where (a) is attributed to the RIP~\eqref{rip}. With $ \delta_{2 R} \leqslant \delta_{4 R} \leqslant \varepsilon $, we obtained~\eqref{theorem1}.
\end{proof}


\subsection{Enhanced Signal Reconstruction via Off-Grid CS}

Relying on the discrete polar domain sampling, the estimated near-field channel in Turbo-CoSaMP can be regarded as a linear combination of paths from several fixed positions (i.e., an ``on-grid'' estimation). This allows us to gain a quick knowledge of the dominant channel components with low computational complexity. However, no matter how fine we draw the grid, the mismatch still exists between the continuously distributed physical path and the discrete virtual path. This so-called off-grid or basis mismatch effect leads to the unpredictable energy spread between sampled array responses, which would cause the degeneration of the proposed on-grid reconstruction algorithm.

To overcome such limitation, parameters related to the near-field channel, i.e., the DoA $ \theta $ and the distance $ d $, should be estimated over the continuous space. Recall that for each atom of the form $ \mathbf{a}_{j} \mathbf{e}_{\mathrm{n}}^{T}( \theta, d ) $, the greedy-pursuit-based algorithm aims at minimizing the residual power $ \| \mathbf{R} - g \mathbf{a}_{j} \mathbf{e}_{\mathrm{n}}^{T}( \theta, d ) \|_{2}^{2} $ though the channel gain $ g $. This is equivalent to minimizing the following function
\begin{align}
	E( g, j, \theta, d ) = 2 \Re \left\lbrace g \mathbf{a}_{j}^{H} \mathbf{R} \mathbf{e}_{\mathrm{n}}^{*}( \theta, d ) \right\rbrace - | g |^{2} \| \mathbf{e}_{\mathrm{n}}( \theta, d ) \|_{2}^{2}.
\end{align}
However, considering the continuous nature of $ ( \theta, d ) $, it is prohibitive to exhaustively search the optimal $ ( g, j, \theta, d ) $ over the whole feasible region. Fortunately, the proposed Turbo-CoSaMP algorithm has already provided the active codeword index $ \hat{j} $ and a rough estimate $ ( \hat{\theta}, \hat{d} ) $ via a grid-based search. Also, the corresponding gain coefficient $ \hat{g} $ can be obtained by a LS estimation. Based on such a starting point, we adopt a Newton update step to emulate the pursuit over the continuum, expressed as
\begin{align}
	\left[ \hat{\theta}^{\prime}, \hat{d}^{\prime} \right]^{T} = \left[ \hat{\theta}, \hat{d} \right]^{T} - \ddot{\mathbf{E}}( \hat{g}, \hat{j}, \hat{\theta}, \hat{d} )^{-1} \dot{\mathbf{E}}( \hat{g}, \hat{j}, \hat{\theta}, \hat{d} ), \label{nstep}
\end{align}
where $ \dot{\mathbf{E}}( \hat{g}, \hat{j}, \hat{\theta}, \hat{d} ) $ is the Jacobian matrix expressed as
\begin{align}
	\dot{\mathbf{E}}( \hat{g}, \hat{j}, \hat{\theta}, \hat{d} ) = \left[ \frac{\partial E}{\partial \hat{\theta}}, \frac{\partial E}{\partial \hat{d}} \right]^{T},
\end{align}
where
\begin{align}
	\frac{\partial \mathbf{E}}{\partial \bar{x}} = 2 \Re \left\{ \hat{g} \left( \mathbf{R}^{T} \mathbf{a}_{\hat{j}}^{*} - \hat{g} \mathbf{e}_{\mathrm{n}}( \hat{\theta}, \hat{d} ) \right)^{H} \frac{\partial \mathbf{e}_{\mathrm{n}}( \hat{\theta}, \hat{d} )}{\partial \bar{x}} \right\},
\end{align}
with $ \bar{x} \in \{ \hat{\theta}, \hat{d} \} $,
and $ \ddot{\mathbf{E}}( \hat{g}, \hat{j}, \hat{\theta}, \hat{d} ) $ is the Hessian matrix expressed as
\begin{align}
	\ddot{\mathbf{E}}( \hat{g}, \hat{j}, \hat{\theta}, \hat{d} ) = \begin{bmatrix} \frac{\partial^{2} E}{\partial \hat{\theta}^{2}} & \frac{\partial^{2} E}{\partial \hat{\theta} \partial \hat{d}} \\\frac{\partial^{2} E}{\partial \hat{d} \partial \hat{\theta}} & \frac{\partial^{2} E}{\partial \hat{d}^{2}} \end{bmatrix},
\end{align}
where
\begin{align}
	\frac{\partial^{2} \mathbf{E}}{\partial \bar{x}_{1} \partial \bar{x}_{2}} = 2 \Re \left\{ \hat{g} \left( \mathbf{R}^{T} \mathbf{a}_{\hat{j}}^{*} - \hat{g} \mathbf{e}_{\mathrm{n}}( \hat{\theta}, \hat{d} ) \right)^{H} \frac{\partial^{2} \mathbf{e}_{\mathrm{n}}( \hat{\theta}, \hat{d} )}{\partial \bar{x}_{1} \partial \bar{x}_{2}} - | \hat{g} |^{2} \frac{\partial \mathbf{e}_{\mathrm{n}}^{H}( \hat{\theta}, \hat{d} )}{\partial \bar{x}_{2}} \frac{\partial \mathbf{e}_{\mathrm{n}}( \hat{\theta}, \hat{d} )}{\partial \bar{x}_{1}} \right\},
\end{align}
with $ \bar{x}_{1}, \bar{x}_{2} \in \{ \hat{\theta}, \hat{d} \} $. The proceeding step~\eqref{nstep} is repeated for $ T_{\mathrm{N}} $ times, which can be regarded as a coordinate descent optimization over a continuum. We accept the update step if the function $ E $ is locally concave, i.e., $ \operatorname{det} ( \ddot{\mathbf{E}} ) > 0 $ and the $ ( 1, 1 ) $-th element of $ \ddot{\mathbf{E}} $ is lower than $ 0 $. It also requires that the cost $ \| \mathbf{a}_{\hat{j}}^{H} \mathbf{R} \mathbf{e}_{\mathrm{n}}^{*}( \hat{\theta}', \hat{d}' ) \|_{2}^{2} > \| \mathbf{a}_{\hat{j}}^{H} \mathbf{R} \mathbf{e}_{\mathrm{n}}^{*}( \hat{\theta}, \hat{d} ) \|_{2}^{2} $. This ensures that the refinement leads to the maximum decline of the residual.

\begin{algorithm}[!t]
	\caption{N-Turbo-CoSaMP}
	\label{nturbocosamp}
	\begin{spacing}{1.2}
		\textbf{Input:} Observation $ \mathbf{Y} $, measurement matrices $ \mathbf{A}, \mathbf{B} $, sparsity levels $ K_{\mathrm{a}}, R $, target power $ \tau $, number of Newtonized refinement round $ T_{\mathrm{l}} $ and $ T_{\mathrm{c}} $ \\
		\textbf{Initialize:} $ \widehat{\mathbf{Z}}( 0 ) = \mathbf{0}_{2^{J} \times M} $, row support set $ \mathcal{R}_{\mathrm{r}}( 0 ) = \emptyset $, signal support set $ \mathcal{R}( 0 ) = \emptyset $, residual $ \mathbf{R}( 0 ) = \mathbf{Y} $, $ t = 0 $
		\begin{algorithmic}[1]
			\REPEAT
			\STATE Run lines 1-8, Algorithm~\ref{turbocosamp} to obtain the row support set $ \mathcal{K}_{\mathrm{r}}( t ) $ and the signal support set $ \mathcal{S}( t ) $.
			\STATE Estimate $ \mathbf{g}( t ) = \arg \min \limits_{\mathbf{g}} \| \mathbf{Y} - \sum_{( j, p ) \in \mathcal{S}( t )} g_{j, p} \mathbf{a}_{j} \mathbf{b}_{p}^{T} \|_{F}^{2} $ and let $ \mathcal{P}( t ) = \{ ( g, j, \theta, d ) : ( j, p ) \in \mathcal{S}( t ), \mathbf{e}( \theta, d ) = \mathbf{b}_{p} \} $
			\STATE \textbf{\% Local Newtonized refinement}
			\STATE $ \forall ( g, j, \theta, d ) \in \mathcal{P}( t ) $: update $ ( \theta, d ) $ via the Newton refinement~\eqref{nstep} for $ T_{\mathrm{l}} $ times
			\STATE Merge the parameter set $ \mathcal{K}( t ) = \mathcal{R}( t - 1 ) \cup \mathcal{P}( t ) $
			\STATE \textbf{\% Cyclic Newtonized refinement}
			\FOR{$ t_{\mathrm{c}} = 1, \dots, T_{\mathrm{c}} $}
			\STATE $ \forall ( g', j', \theta', d' ) \in \mathcal{K}( t ) $: update $ ( \theta', d' ) $ via the Newton refinement~\eqref{nstep} based on the residual $ \mathbf{Y} - \sum_{( g, j, \theta, d ) \in \mathcal{K}( t ) \setminus \{ ( g', j', \theta', d' ) \}} g \mathbf{a}_{j} \mathbf{e}_{\mathrm{n}}^{T}( \theta, d ) $
			\ENDFOR
			\STATE LS estimation $ \mathbf{g}'( t ) = \arg \min \limits_{\mathbf{g}'} \| \mathbf{Y} - \sum_{\mathcal{K}( t )} g' \mathbf{a}_{j'} \mathbf{e}_{\mathrm{n}}^{T}( \theta', d' ) \|_{F}^{2} $
			\STATE Screen $ \mathcal{K}_{\mathrm{r}}( t ) $ by retaining the indices of $ K_{\mathrm{a}} $ rows with the largest energy to obtain the row support set $ \mathcal{R}_{\mathrm{r}}( t ) $
			\STATE Screen $ \mathcal{K}( t ) $ by retaining the elements related the to the $ R $ largest $ | g | $ within the rows indexed by $ \mathcal{R}_{\mathrm{r}}( t ) $ to obtain the signal support set $ \mathcal{R}( t ) $
			\STATE $ \forall j \in \mathcal{R}_{\mathrm{r}}( t ) $: $ \widehat{\mathbf{z}}_{j}( t ) = \sum_{( g, j, \theta, d ) \in \mathcal{R}( t )} g \mathbf{a}_{j} \mathbf{e}_{\mathrm{n}}^{T}( \theta, d ) $
			\STATE $ \forall j \in [ 2^{J} ] \setminus \mathcal{R}_{\mathrm{r}}( t ) $: $ \widehat{\mathbf{z}}_{j}( t ) = \mathbf{0} $
			\STATE Update the residual $ \mathbf{R}( t ) = \mathbf{Y} - \sum_{( g, j, \theta, d ) \in \mathcal{R}( t )} g \mathbf{a}_{j} \mathbf{e}_{\mathrm{n}}^{T}( \theta, d ) $
			
			\UNTIL $ \| \mathbf{R}( t ) \|_{F}^{2} \leqslant \tau^{2} $
		\end{algorithmic}
		\textbf{Output:} Signal estimation $ \widehat{\mathbf{Z}}( t ) $
	\end{spacing}
\end{algorithm}

After the newly refined parameters are added to the intermediate set, the Newtonized approach also diverges from the original Turbo-CoSaMP algorithm in introducing a cyclic refinement across the globe set. In particular, the Newtonized coordinate descent step given in~\eqref{nstep} is carried out on each atom based on the residual reached by removing the contribution of the other atoms. Considering that the original channel parameters are locally estimated. Thus, the essence of the cyclic process can be regarded as a feedback mechanism; the newly detected atoms provide information to remove CE errors introduced by previous outcomes and vice versa.

It is worth mentioning that N-Turbo-CoSaMP promotes the sparsity of the near-field channel from a finite domain to a continuous domain capturing the two crucial near-field channel parameters, i.e., the DoA and the distance. The correspondingly proposed Newtonized step in lines 5 and 9, Algorithm~\ref{nturbocosamp}, emulates the pursuit simultaneously over the two-dimensional angle-distance space. This reveals a significant difference from the off-grid estimation in \cite{MRM16} which focuses on only one single target and that in \cite{CD22} where the DoA and the distance are separately refined with a gradient-based local search.

\subsubsection{Complexity and Convergence Analysis}

Compared with Turbo-CoSaMP, the additional refinement steps take the total cost of $ \mathcal{O}( R N M ) $. Thus, the computational complexity of N-Turbo-CoSaMP is on the same level as Turbo-CoSaMP.

Now, we prove the convergence of N-Turbo-CoSaMP. For parameters $ ( g^{t}, j^{t}, \theta^{t}, d^{t} ) $ passed from the $ t $-th iteration, we always have $ \| g^{t} \mathbf{a}_{j^{t}} \mathbf{e}_{\mathrm{n}}^{T}( \theta^{t}, d^{t} ) \|_{F}^{2} \leqslant \| g^{t + 1} \mathbf{a}_{j^{t + 1}} \mathbf{e}_{\mathrm{n}}^{T}( \theta^{t + 1}, d^{t + 1} ) \|_{F}^{2} $ after the cyclic refinement due to the acceptable conditions of the coordinate descent step. Then, CoSaMP retains only $ R $ atoms with the largest gain among recent estimates and previous outcomes. Therefore, $ \| \mathbf{R}( t + 1 ) \|_{F}^{2} < \| \mathbf{R}( t ) \|_{F}^{2} $ always holds. As the residual power is lower bounded, N-Turbo-CoSaMP is guaranteed to converge after sufficient rounds of iteration.

\section{Message Stitching via Channel Clustering}

\subsection{$ K $-Medoids for Clustering}

After retrieving active codewords and channels from all the transmission intervals, the decoder aims to reconstruct the original messages by concatenating the codewords from each transmitter in sequence. Since the codewords across all the intervals are coupled by the similarity of their corresponding channels, we propose to conduct the message stitching process by grouping slot-distributed channels into clusters.

The $ K $-medoids algorithm~\cite{PJ09}, as a partitioning-based clustering method, is well adapted to perform non-hierarchical classification. It proceeds by the following steps.
\begin{itemize}
	\item
	\textbf{Step 1:} Randomly select $ K_{\mathrm{a}} $ objects as the cluster medoids (centers).
	\item
	\textbf{Step 2:} Assign a given data instance to the nearest medoid.
	\item
	\textbf{Step 3:} Update the new medoid of each cluster as the data point that minimizes the total distance to other components in the cluster.
	\item
	Repeat Steps 2, 3 until there is no change in the assignments.
\end{itemize}
Different from the $ K $-means algorithm~\cite{XWA22}, the $ K $-medoids chooses the centrally located data instance instead of the means of all the objects to represent each cluster, making it less sensitive to potential outliers, e.g., contaminated channel vectors caused by codeword collision.

\subsection{Channel Clustering Constraints}

In the practical application of clustering-based decoding, each retrieved channel as a data instance also comes in with some background knowledge, i.e., the sequential order of its belonging interval. However, the conventional $ K $-medoids algorithm fails to explore that information to constrain the cluster placement of the instances. Therefore, we propose to incorporate the background knowledge in the form of instance-level constraints and impose a balanced clustering restriction on the channel assignment step. An unbalanced restriction is further investigated to tackle the situation of codeword reuse.

\subsubsection{Balanced Channel Clustering}

For now, we consider an ideal case that $ K_{\mathrm{a}} $ different codewords are emitted by the $ K_{\mathrm{a}} $ active users at every slot. Intrinsically, each cluster should collect no more than one channel from a given slot. To fulfill such a constraint, we carry out the assignment step of $ K $-medoids on a per-slot basis with $ K_{\mathrm{a}} $ channels allocated to $ K_{\mathrm{a}} $ clusters simultaneously. Suppose that the channels and cluster medoids stand for the factors on two sides of a bipartite graph and the distance between each channel and medoid represents the weight linking the corresponding factors. Then, grouping $ K_{\mathrm{a}} $ channels to $ K_{\mathrm{a}} $ clusters with the minimum intra-group distance is equivalent to the minimum bipartite matching problem, which can be efficiently solved by the Hungarian algorithm~\cite{K55} with computational complexity $ \mathcal{O}( K_{\mathrm{a}}^{3} ) $.

\subsubsection{Unbalanced Channel Clustering Handing Codeword Collision}

URA can be regarded as a competitive process for common codewords. Thus, codeword collision would inevitably happen when two or more users share the same codeword, leading to unbalanced channel-medoid pairs. In the following discussion, we assume that the probability of three or more users choosing the same codeword is negligibly low given a proper codebook size fitting $ K_{\mathrm{a}} $~\cite{SBM21}. Considering a circumstance that only a total number of $ K_{s} < K_{\mathrm{a}} $ codewords are picked by users, the equivalent channel of each reused codeword is the superposition of conflicting users' channels, which generally remains large distances to all the cluster centers. Therefore, we identify the $ K_{\mathrm{a}} - K_{s} $ channels with the largest value of the minimum distance to all the center points as contaminated ones; they are repeatedly added to the channel set for clustering. Thus, channel assignment returns to the balanced case.

\subsection{Overall Algorithm}

\begin{algorithm}[!t]
	\caption{Modified $ K $-Medoids Algorithm for Channel Clustering}
	\label{Kmedoids}
	\begin{spacing}{1.2}
		\textbf{Input:} Data sets $ \{ \widehat{\mathbf{h}}^{( i )}( s ) : i \in [ K_{s} ] \}, s \in [ S ] $
		\begin{algorithmic}[1]
			\STATE Choose $ \widehat{K}_{\mathrm{a}} = \max \{ K_{s} : s \in [ S ] \} $
			\STATE $ \forall k \in [ \widehat{K}_{\mathrm{a}} ] $: Set $ \mathbf{n}_{k}(0, S)  = \widehat{\mathbf{h}}^{( k )}( s ) $ with $ s = \arg \max \limits_{s'} \{ K_{s'} : s' \in [ S ] \} $, $ \mathcal{C}_{k}( 0 ) = \emptyset $
			\STATE $ t = 0 $
			\REPEAT
			\STATE $ t = t + 1 $
			\STATE $ \forall k $: Set $ \mathbf{n}_{k}(t, 0) = \mathbf{n}_{k}(t - 1, S) $
			\FOR{$ s = 1, 2, \dots, S $}
			\STATE Compute cost matrix $ \mathbf{C} $ with $ [ \mathbf{C} ]_{i, k} = \| \widehat{\mathbf{h}}^{( i )}( s ) - \mathbf{n}_{k}(t, s - 1) \|_{2} $
			\IF{$ K_{s} < \widehat{K}_{\mathrm{a}} $}
			\STATE Reshape $ \mathbf{C} $ by adding $ \widehat{K}_{\mathrm{a}} - K_{s} $ rows with the largest minimum distance to all the medoids
			\ENDIF
			\STATE Perform the Hungarian algorithm with input $ \mathbf{C} $ and obtain the output $ \mathbf{V} $
			\STATE $ \forall k $: $ \mathcal{C}_{k}( t ) = \mathcal{C}_{k}( t - 1 ) \cup \{ \sum_{i = 1}^{K_{s}} v_{i, k} \widehat{\mathbf{h}}^{( i )}( s ) \} $
			\STATE $ \forall k $: Update $ \mathbf{n}_{k}(t, s) $ via~\eqref{mupdate}
			\ENDFOR
			\UNTIL{$ \mathcal{C}_{k}( t ) = \mathcal{C}_{k}( t - 1 ) $ for $ \forall k \in [ \widehat{K}_{\mathrm{a}} ] $}
		\end{algorithmic}
		\textbf{Output:} Partitioning of the data set
	\end{spacing}
\end{algorithm}

We denote the set of recovered channels at the $ s $-th slot by $ \{ \widehat{\mathbf{h}}^{( i )}( s ) : i \in [ K_{s} ] \} $ with $ K_{s} = | \mathcal{A}( s ) |_{c} $. $ \mathcal{C}_{k} $ represents the $ k $-th cluster and $ \mathbf{n}_{k} \in \mathbb{C}^{P} $ is the medoid of $ \mathcal{C}_{k} $. If the value of $ K_{\mathrm{a}} $ is unknown, its estimation $ \widehat{K}_{\mathrm{a}} $ can be decided by the largest value of $ K_{s}, s \in [ S ] $ and the channels from the corresponding slot are chosen as the initial cluster centers. In order to capture the element-wise difference between two vectors, we weight the channel-medoid distance by the Euclidean distance and then calculate a cost matrix $ \mathbf{C} \in \mathbb{C}^{K_{s} \times \widehat{K}_{\mathrm{a}}} $ with the $ (i,k) $-th entry $ [ \mathbf{C} ]_{i, k} = \| \widehat{\mathbf{h}}^{( i )}( s ) - \mathbf{n}_{k} \|_{2} $.

If $ K_{s} = \widehat{K}_{\mathrm{a}} $ is satisfied at the current slot, a balanced channel clustering is performed by solving the Hungarian algorithm with input $ \mathbf{C} $ a square matrix. The algorithm output $ \mathbf{V} \in \{0, 1\}^{K_{s} \times \widehat{K}_{\mathrm{a}}} $ is a binary matrix with exactly one nonzero element in each row and each column and the $ ( i, k ) $-th entry $ v_{i, k} = 1 $ implies that the $ i $-th channel belongs to the $ k $-th cluster. With each cluster allotted as $ \mathcal{C}_{k} = \mathcal{C}_{k} \cup \{ \sum_{i = 1}^{K_{s}} v_{i, k} \widehat{\mathbf{h}}^{( i )}( s ) \} $, the corresponding medoid is renewed as the centrally located component, i.e.,
\begin{align}
	\mathbf{n}_{k} = \arg \min \limits_{\widehat{\mathbf{h}}^{( i )}( s ) \in \mathcal{C}_{k}} \sum_{( i', s' ) \in \mathcal{I} \setminus \{ ( i, s ) \}} \left\| \widehat{\mathbf{h}}^{( i )}( s ) - \widehat{\mathbf{h}}^{( i' )}( s' ) \right\|_{2}, \label{mupdate}
\end{align}
where the set $ \mathcal{I} = \{ ( i', s' ) : \widehat{\mathbf{h}}^{( i' )}( s' ) \in \mathcal{C}_{k} \} $.

The unbalanced channel clustering takes a similar process, while the difference lies in that we shape $ \mathbf{C} $ to a square matrix by filling in $ \widehat{K}_{\mathrm{a}} - K_{s} $ rows with the largest minimum distance to all the medoids. Accordingly, the output of the Hungarian algorithm is resized by adding these rows to their original locations.

The overall algorithm is summarized in Algorithm~\ref{Kmedoids}, whose computational complexity is in the order of $ \mathcal{O}( S K_{\mathrm{a}}^{2} ( M + K_{\mathrm{a}} ) ) $. Note that the contaminated channels, caused by codeword collisions, would have almost no effect on the medoid update step as the $ K $-medoids determines each cluster center by the central component as in line~13, Algorithm~\ref{Kmedoids}, instead of the centroid of all the components. Also, the channel features of conflicting users, i.e., whether they are separable with distinctive DoAs and distances, have no effect on the method itself. We would like to mention that~\cite{XWA22} studies the message stitching process for URA regarding the far-field multipath channel, where $ K $-means is introduced for angular domain channel clustering. However, each data instance, even for the contaminated channels, should be included to update the cluster center as the centroid. To suppress the influence of ill-conditioned channels, one would have to cancel the interference of conflicting users' channels, which highly relies on the separability of their own DoAs.

\section{Simulation Results}

In this section, we provide simulation results to verify the performance of the proposed JADCE algorithms and uncoupled URA scheme for near-field communications. The considered carrier frequency in the system is $ 3 $ GHz and the BS is equipped with $ M = 128 $ antennas unless otherwise mentioned. $ K_{\mathrm{a}} = 100 $ active near-field users have their distances to the BS uniformly lie in the range of $ ( 30, 100 ) $ meters and DoAs uniformly distributed in $ ( -1, 1 ) $. We consider a non-line-of-sight (NLOS) propagation environment for the near-field channel with $ L = 2 $ paths~\cite{HJW20} and each path gain is generated independently from a complex normal distribution~\cite{KGW20}. The codebook $ \mathbf{A} $ is generated as a partial DFT matrix with $ J = 14 $. The received signal-to-noise ratio (SNR) is defined as $ \mathrm{SNR} = \| \mathbf{AXB}^{T} \|_{F}^{2} / \| \mathbf{W} \|_{F}^{2} $. In URA, the error event probability is described as the sum of the per-user probability of misdetection and the probability of false-alarm, expressed as
\begin{align}
	P_{\mathrm{e}} = \dfrac{1}{K_{\mathrm{a}}} \sum_{k \in [ K_{\mathrm{a}} ]} \mathbb{P} \left( \mathbf{m}_{k} \notin \mathcal{L} \right) + \dfrac{\left| \mathcal{L} \backslash \{ \mathbf{m}_{k}: k \in [ K_{\mathrm{a}} ] \} \right|_{c} }{\left| \mathcal{L} \right|_{c} },
\end{align}
where $ \mathcal{L} $ is the decoded message list. We measure the CE accuracy by the NMSE with respect to the spatial domain channel, expressed as
\begin{align}
	\mathrm{NMSE} = \dfrac{\sum_{j \in \mathcal{A}_{s}} \| \widehat{\mathbf{z}}_{j, :} - \mathbf{z}_{j, :} \|_{2}^{2}}{\sum_{j \in \mathcal{A}_{s}} \| \mathbf{z}_{j, :} \|_{2}^{2}}, s \in [ S ].
\end{align}
Furthermore, for the proposed CoSaMP-based algorithms, the sparsity levels is determined as $ R = 400 $. We set $ T_{\mathrm{l}} = T_{\mathrm{c}} = 3 $ for the Newtonized refinement and $ \tau^{2} = \frac{\| \mathbf{Y} \|_{F}^{2}}{5 ( \mathrm{SNR} + 1 )} $ for the stopping criterion.

\subsection{Performance of the Near-Field Channel Sampling Method}

We consider various channel sampling methods to generate the matrix $ \mathbf{B} $:
\begin{itemize}
	\item
	The angular domain sampling by the spatial Fourier transform, i.e., $ \mathbf{B} = \mathbf{U} $ as in~\eqref{dftmtx}.
	\item
	The polar domain sampling developed in~\cite{CD22} where DoAs and distances are sampled separately. We set the minimum distance $ \rho_{\mathrm{min}} = D_{F} $ and the control parameter $ \beta_{\Delta} = 1.2 $, as specified in~\cite{CD22}, leading to a maximum coherence of value $ 0.7908 $ between columns. We denote the corresponding channel dictionary by $ \mathbf{B}_{\beta} $, which has dimensions $ 128 \times 768 $.
	\item
	The proposed polar domain sampling where the intervals of DoAs and distance rings are jointly decided based on the maximum permissible coherence level. For this method, we fix $ \gamma = 0.5816 $, resulting in a channel dictionary $ \mathbf{B}_{\gamma} $ with dimensions $ 128 \times 741 $.
\end{itemize}

\begin{figure}[t]
	\centering
	\includegraphics[width=7cm]{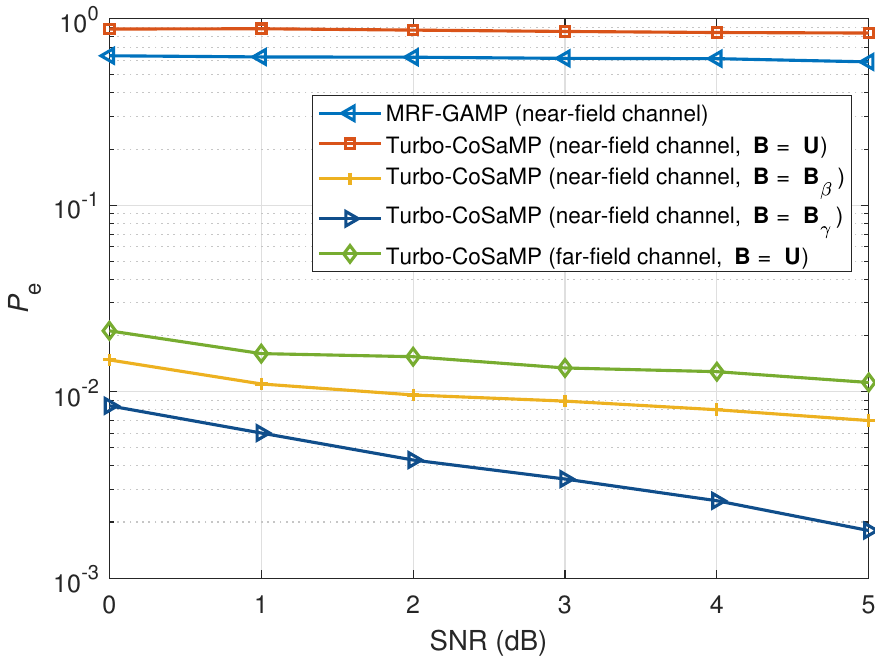}
	\caption{Detection error performance of various channel sampling algorithms versus SNR with $ K_{\mathrm{a}} = 100 $, $ M = 128 $, and $ N = 50 $.}
	\label{sampcomp}
\end{figure}

Under the near-field channel settings, we investigate how these methods would affect the AD performance $ P_{\mathrm{e}} $ within one transmission slot using the proposed Turbo-CoSaMP algorithm. We also present the generalized AMP algorithm with a Markov random field support structure (MRF-GAMP)~\cite{XWA22} tailored for the angular domain CE. For a fair comparison, we employ $ \mathbf{A} $ as an i.i.d. Gaussian matrix for the AMP-based algorithm. The detection error probabilities with $ N = 50 $ are shown in Fig.~\ref{sampcomp}. The failure in AD of the matching-pursuit-based methods suggests that the near-field channel feature cannot be fully captured by the $ L $ dominant coefficients of its angular domain representation due to the severe energy spread. As the density of nonzero elements significantly increases in the channel vector obtained through the spatial Fourier transform, the MRF-GAMP algorithm, which exploits the angular domain structure, performs inadequately under near-field conditions with limited measurements. Moreover, the proposed polar domain sampling method demonstrates superior performance compared to the one investigated in~\cite{CD22}, given a similar number of channel samples. This can be attributed to the lower column coherence exhibited by the former channel dictionary, which enables precise identification of active elements through correlation calculation $ \mathbf{A}_{\mathcal{K}_{\mathrm{r}}}^{H} \mathbf{RB}^{*} $ in CoSaMP, resulting in enhanced accuracy in AD.

We also present in Fig.~\ref{sampcomp} investigating Turbo-CoSaMP working in the far-field region with the angular domain sampling. It is found that the near-field-suited version exploiting polar domain sparsity surpasses the counterpart considering the angular domain sparsity in the far-field. This can be explained by re-examining the equivalent SMV model that the sensing matrix $ \boldsymbol{A} = \mathbf{B} \otimes \mathbf{A} $ offers $ 2^{J} P $ measurements for the former while $ 2^{J} M < 2^{J} P $ for the latter.

\subsection{Performance of Proposed CS algorithms}

\begin{figure*}[t]
	\centering
	\subfigure[]{\includegraphics[width=5cm]{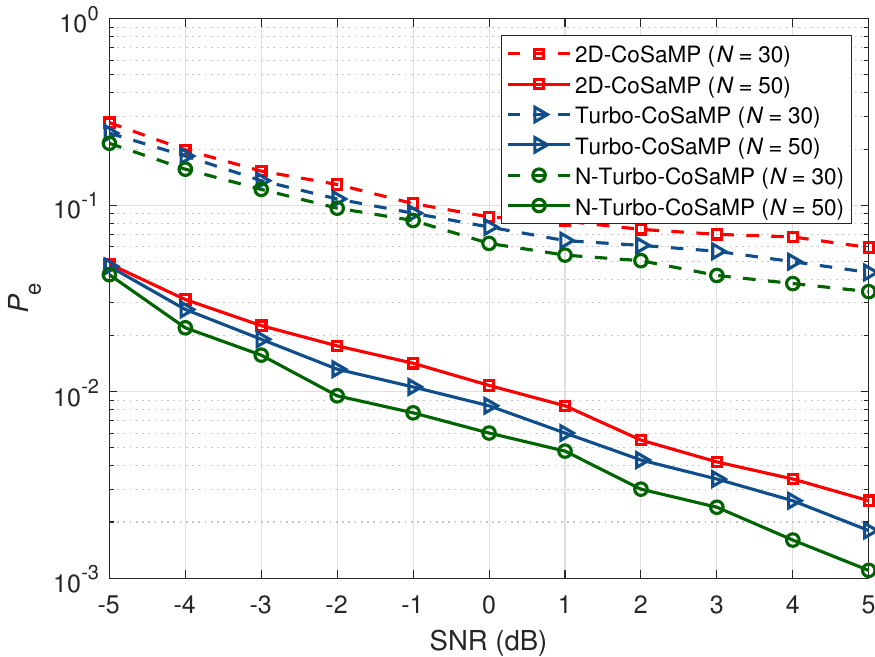} \label{cscompad}} \hspace{0.2cm}
	\subfigure[]{\includegraphics[width=5cm]{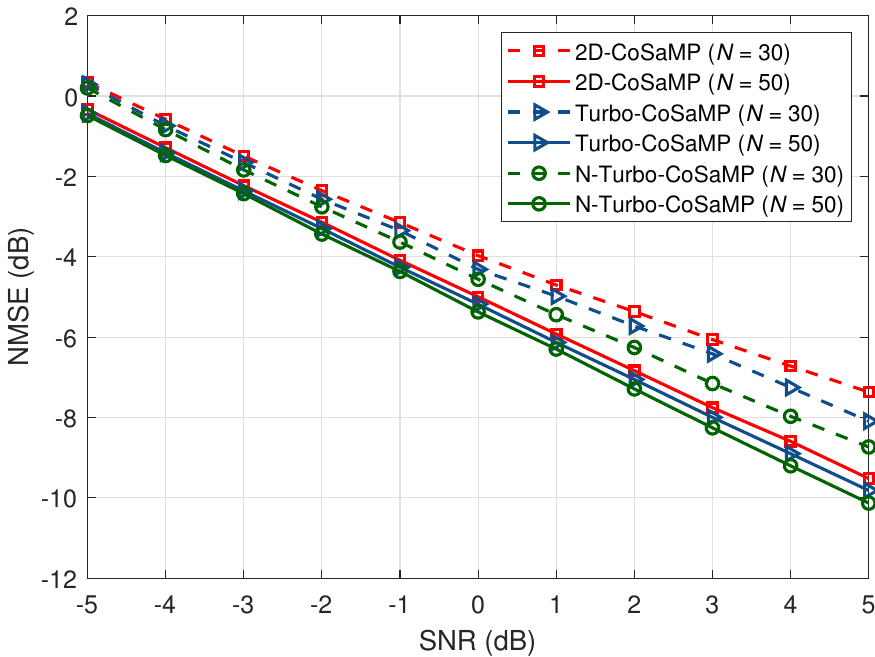} \label{cscompce}} \hspace{0.2cm}
	\subfigure[]{\includegraphics[width=5cm]{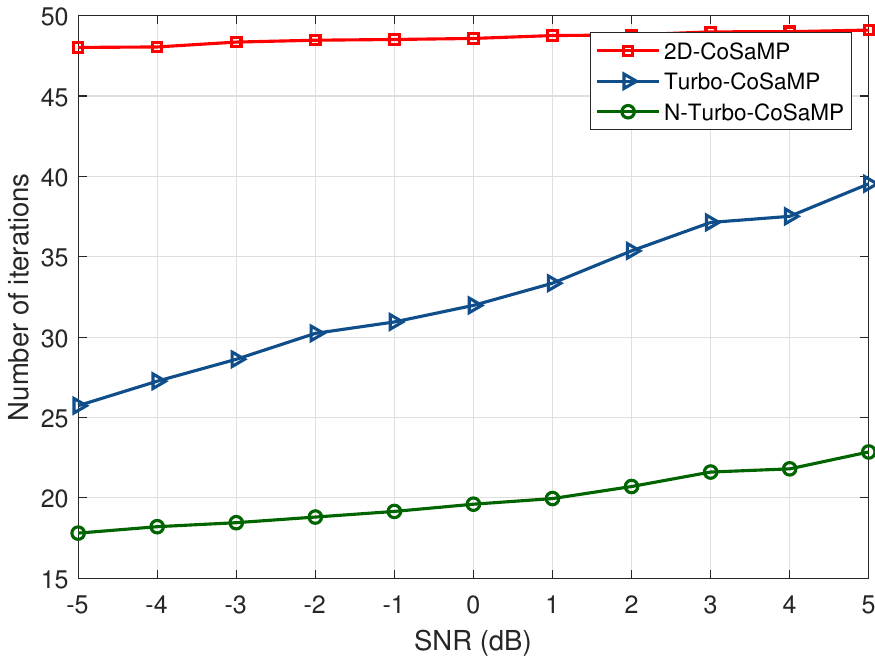} \label{cscompit}}
	\caption{JADCE performance of proposed algorithms with $ K_{\mathrm{a}} = 100 $ and $ M = 128 $. (a) Detection error rate versus SNR under different values of $ N $; (b) NMSE versus SNR under different values of $ N $; (c) Average number of iterations until convergence versus SNR with $ N = 50 $.}
	\label{cscomp}
\end{figure*}

We first examine the JADCE performance of the proposed CS algorithms within each transmission slot. The AD and CE qualities of Turbo-CoSaMP and N-Turbo-CoSaMP under different values of block length and SNR are shown in Fig.~\ref{cscompad} and Fig.~\ref{cscompce}, respectively, where we also operate 2D-CoSaMP on the JADCE model~\eqref{mtxsig} to form a benchmark. It can be seen that the error measures decrease when the length of the coherence block increases. In fact, the increase in the number of measurements reduces the value of the isometry constant of the measurement matrices, leading to a more accurate signal proxy by calculating the correlation. Note that the proposed methods perform effectively in the region where the block length is less than the number of active users, which is the advantage brought by leveraging the polar domain sparsity as fewer channel coefficients are required for reconstruction. Also, we observe that the JADCE performance improves when the value of SNR increases. This finding can be supported by the description in Remark 1 that the NMSE of the signal estimation is proportional to the noise variance. Furthermore, Turbo-CoSaMP outperforms 2D-CoSaMP since S-CoSaMP can confine possible locations of the active elements by considering the spatial domain row sparsity in the first place. Besides, the off-grid N-Turbo-CoSaMP algorithm restrains the basis mismatch influence of the grid-based searching, which leads to a performance gain compared to the on-grid Turbo-CoSaMP. Finally, the average number of iterations until convergence for both algorithms is presented in Fig.~\ref{cscompit}. It can be inferred that the convergence rates of the proposed methods are remarkably faster than OMP-based algorithms which need at least $ R $ iterations to recover $ R $ elements. Also, the off-grid algorithm, accelerated by the Newtonized coordinate descent, pays fewer rounds of iterations to converge than the on-grid one.

\begin{figure*}[t]
	\centering
	\begin{minipage}[b]{5cm}
		\subfigure[$ \mathrm{SNR} = 0 $ dB]{\includegraphics[width=5cm]{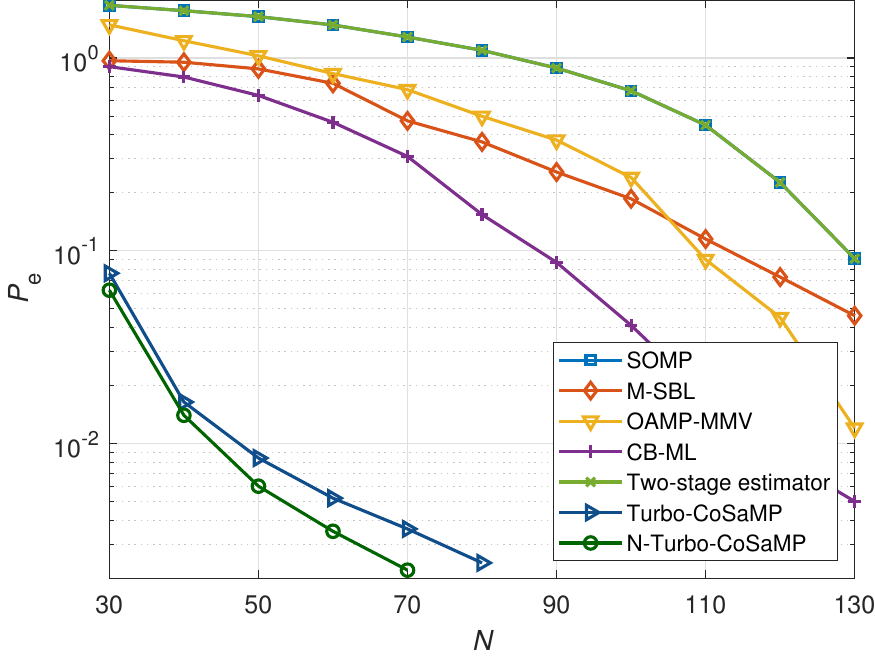} \label{adcomp}} \\
		\subfigure[$ \mathrm{SNR} = 0 $ dB]{\includegraphics[width=5cm]{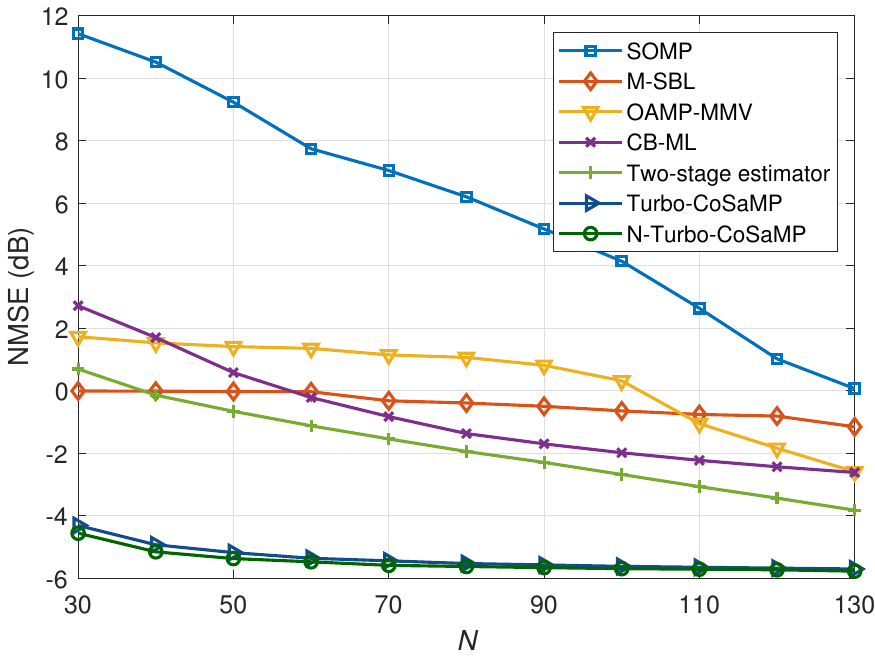} \label{cecomp}}
	\end{minipage} \hspace{0.2cm}
	\begin{minipage}[b]{5cm}
		\subfigure[$ N = 50 $]{\includegraphics[width=5cm]{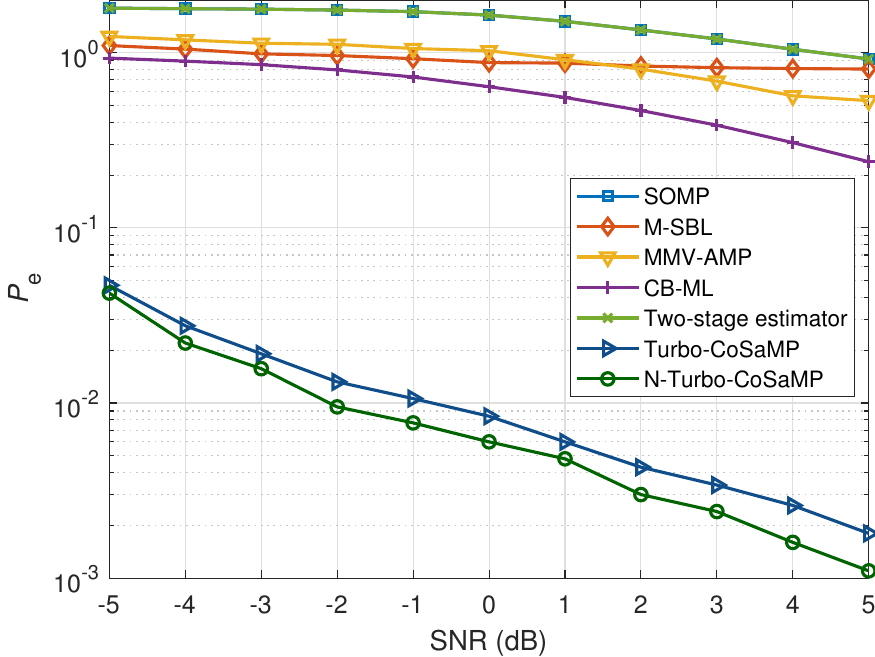} \label{adcomp1}} \\
		\subfigure[$ N = 50 $]{\includegraphics[width=5cm]{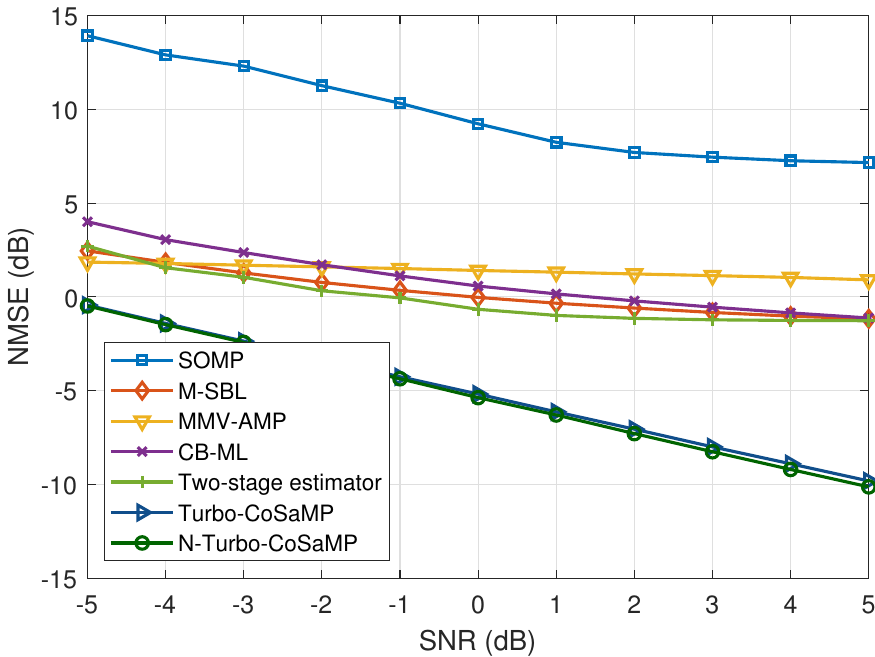} \label{cecomp1}}
	\end{minipage} \hspace{0.2cm}
	\begin{minipage}[b]{5cm}
		\subfigure[$ N = 100 $]{\includegraphics[width=5cm]{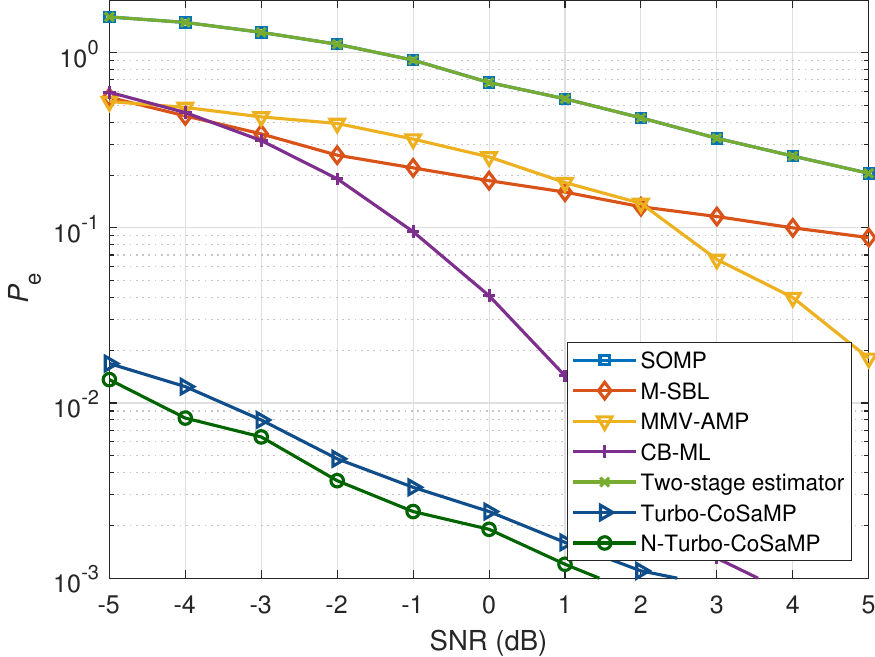} \label{adcomp2}} \\
		\subfigure[$ N = 100 $]{\includegraphics[width=5cm]{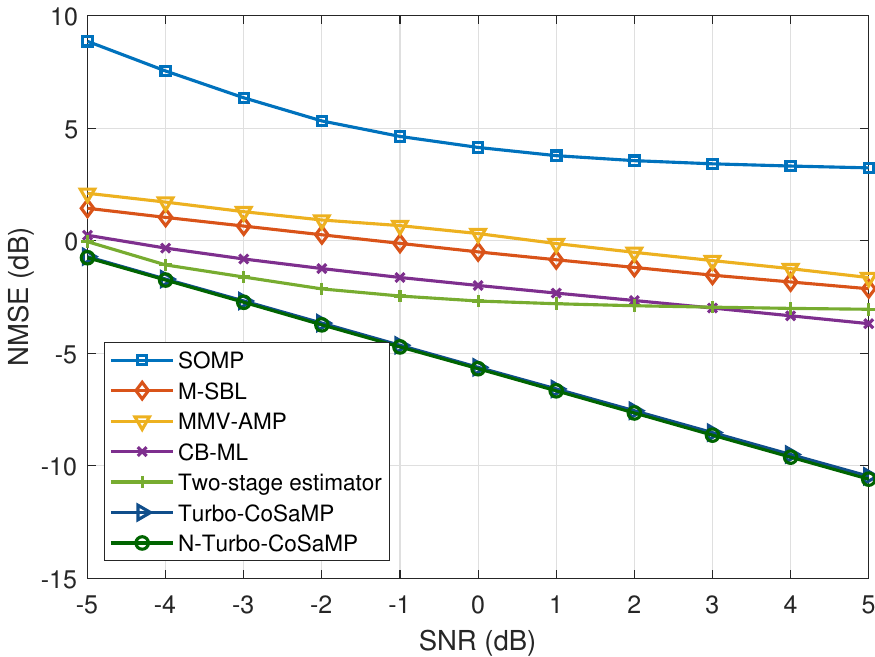} \label{cecomp2}}
	\end{minipage}
	\caption{JADCE performance of various algorithms with $ K_{\mathrm{a}} = 100 $, $ M = 128 $; (a) Detection error rate versus block length $ N $ when $ \mathrm{SNR} = 0 $ dB; (b) NMSE versus block length $ N $ when $ \mathrm{SNR} = 0 $ dB; (c) Detection error rate versus SNR when $ N = 50 $; (d) NMSE versus SNR when $ N = 50 $; (e) Detection error rate versus SNR when $ N = 100 $; (d) NMSE versus SNR when $ N = 100 $.}
	\label{jadcecomp}
\end{figure*}

In Fig.~\ref{jadcecomp}, we compare the JADCE performance of the proposed algorithms with various sparse recovery algorithms, including S-OMP~\cite{TGS06}, multiple sparse Bayesian learning (M-SBL)~\cite{WR07}, and the covariance-based maximum likelihood estimator (CB-ML)~\cite{FHJ21}. In addition, we apply OAMP-MMV~\cite{CLP21} for comparison, considering that it is more robust to potentially correlated or ill-conditioned sensing matrices. We further develop a two-stage estimator, which first recovers the spatial domain channel $ \widehat{\mathbf{z}}_{k, :}^{T} $ by S-OMP and then reconstructs each polar domain channel $ \widehat{\mathbf{x}}_{k} $ from $ \widehat{\mathbf{z}}_{k, :}^{T} = \mathbf{B} \widehat{\mathbf{x}}_{k} $ using OMP. It is seen in Fig.~\ref{adcomp} that compared to the proposed sparsity-exploiting algorithms that apply to the scenarios where $ N \ll K_{\mathrm{a}} $, S-OMP, M-SBL, and OAMP-MMV are regulated to operate in the region where $ N \gtrsim K_{\mathrm{a}} $ due to the fundamental limit of MMV approaches that the number of successfully retrieved active users grows with a scale of $ \mathcal{O}( N / \log ( 2^{J} ) ) $. Furthermore, we investigate the AD performance of these algorithms varying from different values of SNR under short and long block lengths in Fig.~\ref{adcomp1} and Fig.~\ref{adcomp2}, respectively. It is also shown that the proposed Turbo-CoSaMP and N-Turbo-CoSaMP algorithms outperform the MMV methods on both circumstances, which attributes to the exploitation of the spare structure of the near-field channel in the polar domain. On the other hand, although~\cite{FHJ21} derives that CB-ML yields a better scaling law for AD than that of AMP, it misbehaves in the considered near-field multipath channels since it is developed upon an independent channel model with no antenna correlation. Finally, by comparing the two-stage estimator with S-OMP in Fig.~\ref{cecomp}, Fig.~\ref{cecomp1} and Fig.~\ref{cecomp2}, we find that capturing the prominent elements in the sparse polar domain channel significantly improves the CE performance. Even so, it pays a cost of $ 100 $ block length to achieve similar results of Turbo-CoSaMP.

\subsection{Performance of Uncoupled URA Scheme for Near-Field Users}

In the designed uncoupled URA scheme, each active user arranges to send a $ B = 96 $-bit message; it is split into $ S = 7 $ sub-messages of length $ J = 14 $. The corresponding UCS schemes working in the near-field with polar domain sampling are compared in Fig.~\ref{uracomp} under different choices of SNR and block length. We observe that UCS with the N-Turbo-CoSaMP-based inner decoder outperforms that with Turbo-CoSaMP, which is directly attributed to the improvement in JADCE as revealed in Fig.~\ref{cscomp}. Moreover, we employ Turbo-CoSaMP with angular domain sampling (i.e., letting $ \mathbf{B} $ be a DFT matrix) as the CS decoder in UCS under both near-field and far-field conditions. The error performance of the considered UCS schemes is also given in Fig.~\ref{uracomp}. As indicated in Section VI-A, the overcomplete matrix for polar domain sampling offers more measurements to improve the JADCE performance, leading to our finding that the designed UCS scheme for near-field communications surpasses its counterpart in the far-field.

\begin{figure}[t]
	\centering
	\includegraphics[width=7cm]{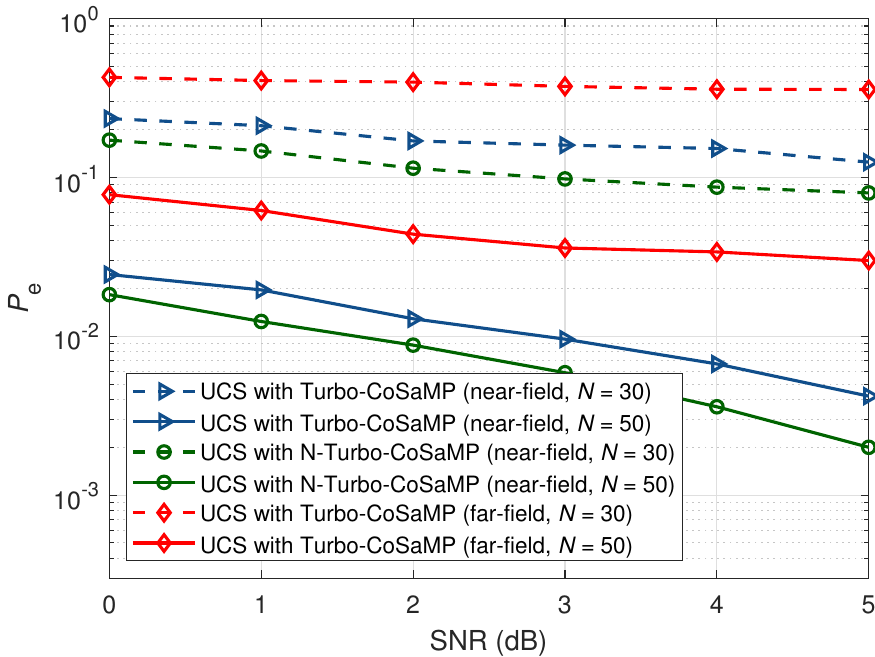}
	\caption{Decoding performance of various URA schemes under different values of $ N $ with $ K_{\mathrm{a}} = 100 $ and $ M = 128 $.}
	\label{uracomp}
\end{figure}

We consider the following multi-slot MIMO URA schemes for comparison. These schemes adopt the common codebook of size $ 50 \times 2^{14} $ and the same channel setting described before.
\begin{itemize}
	\item
	\textbf{Scheme 1:} This is the CCS-based approach in~\cite{FHJ21} with a CB-ML-cored inner decoder for AD. We divide the message into $ S_{\mathrm{ccs}} = 25 $ segments of length $ J = 14 $ with the profile $ [ 14, 4, \dots, 4, 0, 0, 0 ] $. To fill each segment, we append the parity check bits produced by the pseudo-random linear combinations of the data bits from the previous segments.
	\item
	\textbf{Scheme 2:} The massive MIMO URA scheme in~\cite{LW21} where S-OMP acts as the CS decoder. A similar two-stage estimator is applied in~\cite{LW21} to enhance the CE. We replace the corresponding channel compression dictionary with the polar domain sampling matrix $ \mathbf{B} $ to suit the near-field case.
	\item
	\textbf{Scheme 3:} The UCS scheme investigated in~\cite{XWA22} which leverages the angular domain sparsity for JADCE and CSI-aided clustering decoding. We set $ \mathbf{A} $ as an i.i.d. Gaussian matrix here for a fair comparison.
\end{itemize}
The detection error rates of these approaches are provided in Fig.~\ref{ucscomp}. The proposed UCS scheme achieves the lowest error rate due to the superiority of the proposed CoSaMP-based algorithms leveraging polar domain sparsity, with the evidence revealed in Fig.~\ref{jadcecomp}. The URA scheme designed in~\cite{FHJ21} operates with the consideration of an i.i.d. Rayleigh fading massive MIMO channel model and the approach in~\cite{XWA22} relies highly on the angular domain sparsity in the far-field situation. They both have certain limitations when adapting to near-field communications with antenna-correlated channels and severe energy spread encountered in the angular domain transformation. Although the sparse channel representation is taken for CE enhancement in~\cite{LW21}, AD is still formulated as a MMV problem, asking for a considerably large block length. Furthermore, the spectral efficiency of the proposed uncoupled transmission scheme is $ \frac{B K_{\mathrm{a}}}{S N} = 28 $ bit/channel use and it also outperforms that of the CCS-based scheme, which is $ 9.3 $ bit/channel use.

\begin{figure*}[t]
	\centering
	\subfigure[]{\includegraphics[width=7cm]{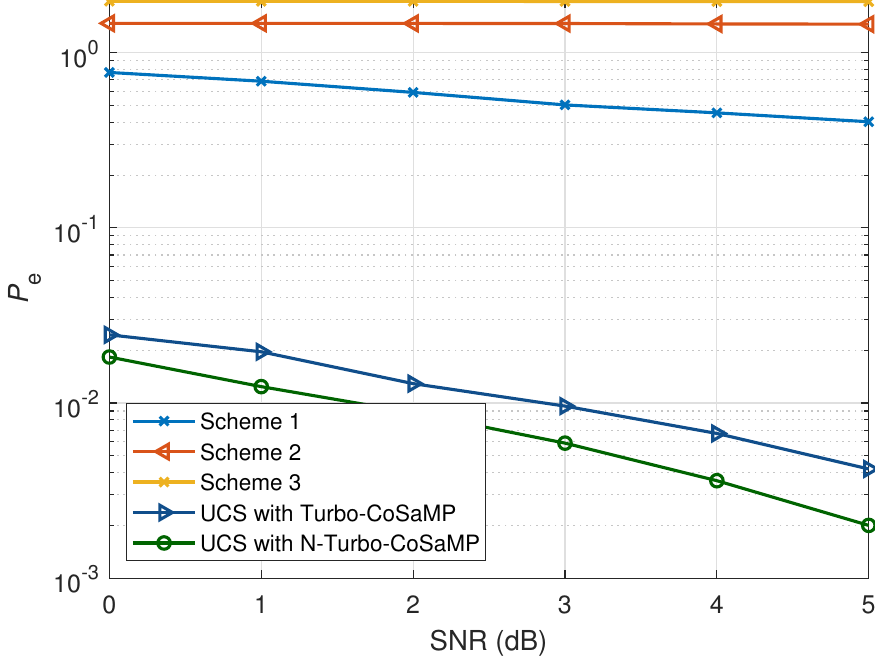} \label{ucscomp}} \hspace{0.5cm}
	\subfigure[]{\includegraphics[width=7cm]{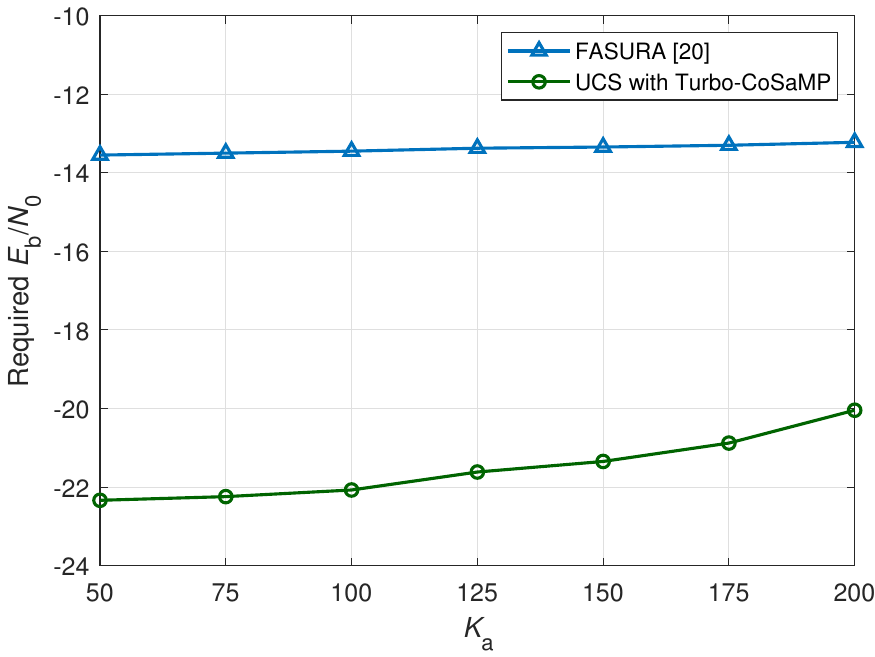} \label{fasuracomp}}
	\caption{Decoding performance of various URA schemes; (a) Decoding error of multi-slot schemes as a function of SNR with $ K_{\mathrm{a}} = 100 $, $ M = 128 $, and $ N = 50 $; (b) Required energy-per-bit to achieve $ P_{\mathrm{e}} < 0.05 $ of UCS and FASURA with $ M = 64 $ and $ N_{\mathrm{tot}} = 3199 $.}
\end{figure*}

Among those tensor-based and pilot-based MIMO URA schemes, we adopt fading spread unsourced random access (FASURA)~\cite{GNC22}, which possesses the best performance, for comparison. For FASURA, the first $ 16 $ bits of the $ 98 $-bit message are used to select a pilot of length $ 895 $ and a spreading sequence of size $ 9 $. The rest $ 82 $ bits are encoded by a cyclic redundancy check (CRC) code and then a polar code. The resultant sequence of length $ 512 $ is modulated to quadrature phase shift keying (QPSK) symbols and spread by the chosen spreading sequences. The total block length is $ N_{\mathrm{tot}} = 895 + \frac{512}{\log_{2} 4} \times 9 = 3199 $. For the proposed UCS with Turbo-CoSaMP, we set $ B = 98 $, $ S = 7 $, $ J = 14 $, and $ N = 457 $, such that $ N_{\mathrm{tot}} = 3199 $ as well. The antenna number is $ M = 64 $, in which case the users within the distance range of $ (10 \ \text{m}, 20 \ \text{m}) $ are still under the near-field condition. Fig.~\ref{fasuracomp} shows the required energy-per-bit $ E_{\mathrm{b}} / N_{0} \triangleq \frac{S \| \mathbf{a}_{j} \|_{2}^{2}}{B \sigma^{2}} $ to achieve the target $ P_{\mathrm{e}} < 0.05 $ of both the approaches. The designed UCS scheme with codeword collision resolution is observed to have a $ 6 \sim 8 $ dB gain when $ K_{\mathrm{a}} $ varies from $ 50 $ to $ 200 $. The energy efficiency of UCS is attributed to its decoding strategy relying on channel reconstruction only. As for FASURA, the polar decoding requires not only the precise JADCE results at the first stage, but also considerable transmit power against the noise and, most importantly, the multi-user interference. 

\subsection{Evaluation of Codeoword Collision Resolution}

\begin{figure*}[t]
	\centering
	\subfigure[]{\includegraphics[width=7cm]{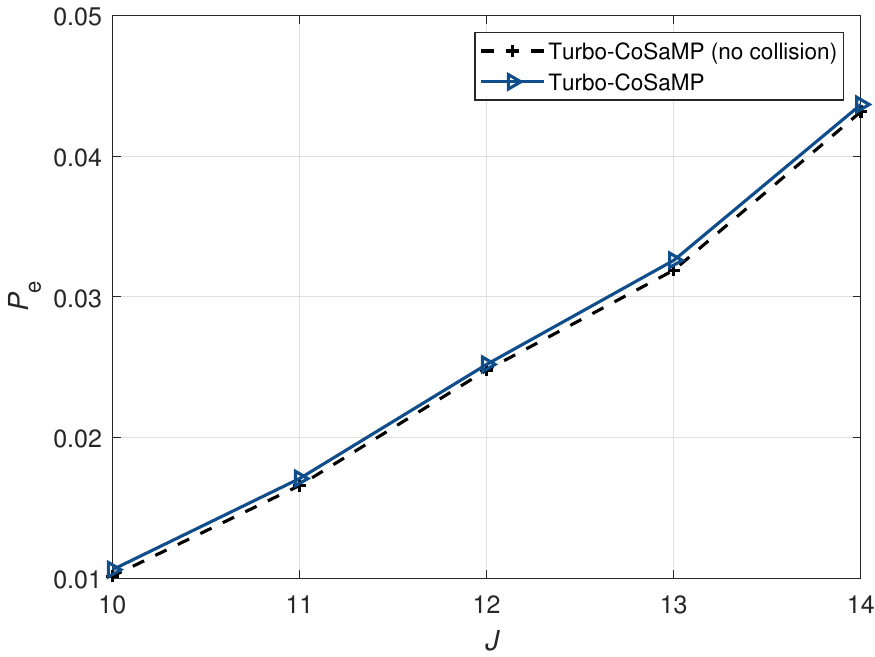} \label{cosampj}} \hspace{0.5cm}
	\subfigure[]{\includegraphics[width=7cm]{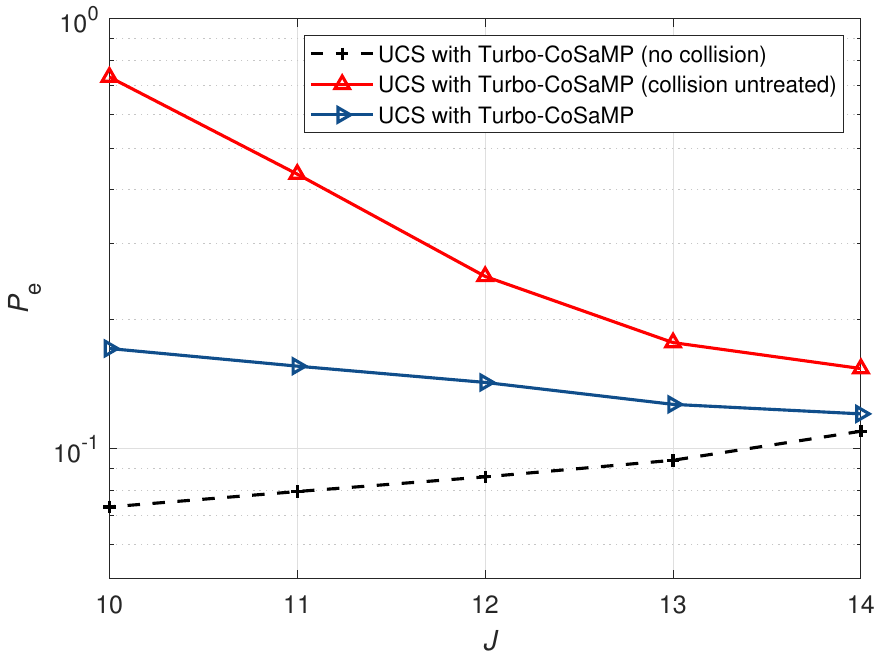} \label{ucsj}}
	\caption{Detection performance against codebook size with $ K_{\mathrm{a}} = 100 $, $ M = 128 $, $ N = 30 $, and $ \mathrm{SNR} = 5 $ dB. (a) AD error rate of Turbo-CoSaMP versus $ J $; (b) Decoding error rate of UCS with Turbo-CoSaMP  versus $ J $.}
	\label{j}
\end{figure*}

In the following simulations, we fix $ K_{\mathrm{a}} = 100 $, $ M = 128 $, $ N = 30 $, and $ \mathrm{SNR} = 5 $ dB. For URA, data transmission takes $ S = 7 $ slots regardless of the codebook size. We set a benchmark by assuming no codeword reuses in each slot. The case where the collisions in clustering decoding are untreated is also considered for comparison. Fig.~\ref{cosampj} shows the detection error of Turbo-CoSaMP against the codebook dimension $ 2^{J} $. As can be seen, the proposed algorithm operates well even under a high chance of collision. In fact, the equivalent channel related to the reused codeword is the sum of conflicting users' channels and usually has relatively large energy. Therefore, codeword collision has limited impacts on the threshold-based AD criterion. It is intriguing to observe that the error rate of AD decreases as the codebook size reduces. This phenomenon can be explained by referring to the RIP provided in Appendix A. In fact, when considering a fixed number of measurements, an increase in the value of $ J $ results in a larger value of $ \delta_{r} $, which subsequently impacts the activity identification step in CoSaMP.

The decoding error of UCS with Turbo-CoSaMP against $ J $ is presented in Fig.~\ref{ucsj}. It can be observed that the event of codeword reuse leads to a significant performance deterioration if not treated. The proposed collision resolution strategy is shown to be effective for reassigning reused codewords to the original sequences, which is another advantage in contrast to the URA schemes designed in~\cite{SBM21,LW21,FMJ22,GNC22} that treat collisions as error events or assume no collision. Furthermore, despite the advantage of a small-sized codebook in AD, the decoding success rate for URA decreases due to unmanageable collisions to a great extent. This is precisely why we opt for $ J = 14 $ in the above simulations.

\section{Conclusion}

In this paper, we presented an URA scheme with massive MIMO to support the massive access of machine-type users in near-field communications. We employed the UCS framework for uplink data transmission using multiple time slots and took two steps to design the decoder. First of all, we proposed an angle-distance sampling method to promote the sparsity of near-field channels in the polar domain. Exploiting both the codeword and the channel sparsity, we put forward the Turbo-CoSaMP algorithm for joint codeword AD and CE at every transmission interval. A Newtonized refinement was further developed within Turbo-CoSaMP to enhance the JADCE performance by harnessing sparsity over the continuum. Afterward, we stitched the slot-distributed codewords in sequence by a clustering decoding process based on a modified $ K $-medoids method. Simulation results have demonstrated that by exploiting the channel sparsity, the proposed sparse recovery methods can achieve low detection error rates even under the circumstance that the available block length is less than the number of active users. On this basis, the corresponding uncoupled URA scheme for near-field communications has proven to be reliable for achieving high spectral efficiency.

\appendices
\section{}

Consider the near-field sampling matrix $ \mathbf{B} $ introduced in Section III, we rewrite \eqref{rip} as
\begin{align}
	\left( 1 - \delta_{r} \right) \| \mathbf{X} \|_{F}^{2} \leqslant \left\| \mathbf{A} \begin{bmatrix} \mathbf{X}_{1} \ldots \mathbf{X}_{P_{\phi}} \end{bmatrix} {\scriptsize \begin{bmatrix} \mathbf{B}_{1}^{T} \\ \vdots \\ \mathbf{B}_{P_{\phi}}^{T} \end{bmatrix}} \right\|_{F}^{2} \leqslant \left( 1 + \delta_{r} \right) \| \mathbf{X} \|_{F}^{2}, \label{rip2}
\end{align}
where $ \mathbf{X}_{p'} $ is a submatrix of $ \mathbf{X} $ constructed by selecting columns with indices from $ ( p' -1 ) P_{\theta} + 1 $ to $ p' P_{\theta} $. We assume that for $ \mathbf{A} $ and each submatrix $ \mathbf{B}_{p'} $ of $ \mathbf{B} $, it holds that
\begin{align}
	C_{p'} ( 1 - \delta_{r_{p'}} ) \| \mathbf{X}_{p'} \|_{F}^{2} \leqslant \| \mathbf{A} \mathbf{X}_{p'} \mathbf{B}_{p'}^{T} \|_{F}^{2} \leqslant C_{p'} ( 1 + \delta_{r_{p'}} ) \| \mathbf{X}_{p'} \|_{F}^{2}.
\end{align}
for some number $ 0 < C_{p'} < 1 $ and $ r_{p'} < r $. Then, \eqref{rip2} is met when $ \sum_{p'} C_{p'} = 1 $ and $ \sum_{p'} r_{p'} < r $. To the worst case, we have that $ C_{p'} = 1 $ and $ \mathbf{X}_{p'} $ is $ r $-sparse. Then, \eqref{rip2} is satisfied if
\begin{align}
	\left( 1 - \delta_{r} \right) \| \mathbf{X}_{p'} \|_{F}^{2} \leqslant \| \mathbf{A} \mathbf{X}_{p'} \mathbf{B}_{r_{p'}}^{T} \|_{F}^{2} \leqslant \left( 1 + \delta_{r} \right) \| \mathbf{X}_{p'} \|_{F}^{2}. \label{rip3}
\end{align}

In the following discussion, we choose $\mathbf{A}$ as a submatrix of a DFT matrix $ \bar{\mathbf{A}} $. As for $ \mathbf{B}_{p'} $, it can also be regarded as a submatrix of $ \bar{\mathbf{B}}_{p'} \in \mathbb{C}^{P_{\theta} \times P_{\theta}} $ whose $ ( p_{1}, p_{2} ) $-th element
\begin{align}
	[ \bar{\mathbf{B}}_{p'} ]_{p_{1}, p_{2}} = \tfrac{1}{\sqrt{P_{\theta}}} \exp \left( j \tfrac{2 \pi}{\lambda_{c}} ( \tilde{d}_{p_{2}, p'} - d_{p_{2}, p', p_{1}} ) \right),
\end{align}
where $ d_{p_{2}, p', p_{1}} = \sqrt{\tilde{d}_{p_{2}, p'}^{2} - \tilde{d}_{p_{2}, p'} \tilde{\theta}_{p_{2}} \kappa_{p_{1}} \lambda_{c}  + \kappa_{p_{1}}^{2} ( \frac{\lambda_{c}}{2} )^{2}} $ and $ \kappa_{p_{1}} = \frac{2 p_{1} - P_{\theta} - 1}{2} $. Denote $ [ \bar{\mathbf{B}}_{p'} ]_{:, p_{2}} $ the $ p_{2} $-th column of $ \bar{\mathbf{B}}_{p'} $. It can be derived that
\small
\begin{align}
	\left( [ \bar{\mathbf{B}}_{p'} ]_{:, p'_{2}} \right)^{H} [ \bar{\mathbf{B}}_{p'} ]_{:, p_{2}} &= \left| \frac{1}{P_{\theta}} \sum_{p_{1} = -( P_{\theta} - 1 ) / 2}^{( P_{\theta} - 1 ) / 2} \exp \left( - j p_{1} \pi (\tilde{\theta}_{p_{2}} - \tilde{\theta}_{p'_{2}}) + j \tfrac{\lambda_{c}}{4} p_{1}^{2} \pi \left( \tfrac{1 - \tilde{\theta}_{p_{2}}^{2}}{\tilde{d}_{p_{2}, p'}} - \tfrac{1 - \tilde{\theta}_{p'_{2}}^{2}}{\tilde{d}_{p'_{2}, p'}} \right) \right) \right| \notag \\
	&\overset{( \mathrm{b} )}{=} \left| \frac{1}{P_{\theta}} \sum_{p_{1} = -( P_{\theta} - 1 ) / 2}^{( P_{\theta} - 1 ) / 2} \exp \left( - j p_{1} \pi (\tilde{\theta}_{p_{2}} - \tilde{\theta}_{p'_{2}}) \right) \right| \notag \\
	&= \left| \frac{\sin \left({\frac{1}{2} P_{\theta} \pi ( \tilde{\theta}_{p_{2}} - \tilde{\theta}_{p'_{2}} )}\right)}{P_{\theta} \sin \left({\frac{1}{2} \pi ( \tilde{\theta}_{p_{2}} - \tilde{\theta}_{p'_{2}} )} \right)} \right| = \begin{cases}
		1, & p_{2} = p'_{2}, \\
		0, & p_{2} \neq p'_{2},
	\end{cases}
\end{align}
\normalsize
where $ ( \mathrm{b} ) $ stands since $ ( 1 - \tilde{\theta}_{p_{2}}^{2} ) / \tilde{d}_{p_{2}, p'} = ( 1 - \tilde{\theta}_{p'_{2}}^{2} ) / \tilde{d}_{p'_{2}, p'} = 1 / \tilde{\phi}_{p'} $. This suggests that $ \bar{\mathbf{B}}_{p'} $ is actually a unitary matrix. Therefore, given the observation $ \mathbf{Y}_{0} = \bar{\mathbf{A}} \mathbf{X}_{0} \bar{\mathbf{B}}_{p'}^{T} $, $ \widehat{\mathbf{X}}_{0} = \bar{\mathbf{A}}^{H} \mathbf{Y}_{0} \bar{\mathbf{B}}_{p'}^{*} $ is the exact estimate of the signal $ \mathbf{X}_{0} $.

Based on the properties of $ \bar{\mathbf{A}} $ and $ \bar{\mathbf{B}}_{p'} $, we are able to present the following corollary measuring how sensing matrices $ \mathbf{A} $ and $ \mathbf{B} $ behave.
\begin{corollary}[Theorem 3.2 in \cite{RV06}]
	Consider the RIP described in \eqref{rip3}, the restricted isometry constant $ \delta_{r} $ of $ \mathbf{A} $ and $ \mathbf{B} $ satisfied $ \mathbb{E}\{ \delta_{r} \} < \varepsilon $ whenever
	\begin{align}
		N M > \dfrac{r \log_{e} 2^{J} P_{\theta}}{\varepsilon^{2}} \log_{e} \left( \dfrac{r \log_{e} 2^{J} P_{\theta}}{\varepsilon^{2}} \right) \log_{e}^{2} r. \label{coro1}
	\end{align}
\end{corollary}
Corollary 1 indicates that increasing the number of measurements $ N $ and the number of antennas $ M $ allows for a lower threshold of the restricted isometry constant. However, we would like to mention that it is challenging to adopt \eqref{coro1} to determine the values of $ N $ and $ M $ based on a desired value of $ \varepsilon $ in the practical application of CS algorithms. There are two reasons:
\begin{enumerate}
	\item
	The RIP is limited to the assumption of exact reconstruction~\cite{RV06}, i.e., the reconstructed signal and the original one are expected to be equal to each other. However, slight deviations in estimation are allowed in the application of CS recovery.
	\item
	We consider the worst case that all the active elements are restrained in a small section of the signal, which is not a common circumstance in practice.
\end{enumerate}

\bibliographystyle{IEEEtran}
\bibliography{reference}

\end{document}